\documentclass[a4paper,UKenglish,cleveref, autoref, thm-restate]{lipics-v2021}
\usepackage[pdftex]{xcolor}
\usepackage{amsmath}
\usepackage{url}
\usepackage{amsthm}
\usepackage{multibib}
\usepackage{cleveref}
\usepackage{newtxtext}
\usepackage[varg]{newtxmath}
\usepackage{algorithm, algorithmic}
\usepackage{tikz}
\usetikzlibrary{decorations.pathreplacing}
\usepackage{bm}
\usepackage{makecell}
\usepackage{todonotes}

\usepackage{complexity}

\usetikzlibrary{calc}

\title{Finding Shortest Reconfiguration Sequences on Independent Set Polytopes}%

\author{Jean Cardinal}
{Université libre de Bruxelles (ULB), Brussels, Belgium}
{jean.cardinal@ulb.be}
{https://orcid.org/0000-0002-2312-0967} 
{}

\author{Kevin Mann}
{Universit\"at Trier, Fachbereich IV, Informatikwissenschaften, Germany}
{mann@uni-trier.de}
{https://orcid.org/0000-0002-0880-2513} 
{}

\author{Akira Suzuki}
{Center for Data-driven Science and Artificial Intelligence, Tohoku University, Sendai, Japan}
{akira@tohoku.ac.jp}
{https://orcid.org/0000-0002-5212-0202} 
{}

\author{Takahiro Suzuki}
{Graduate School of Information Sciences, Tohoku University, Sendai, Japan}
{takahiro.suzuki.q4@dc.tohoku.ac.jp}
{https://orcid.org/0009-0005-8433-3789} 
{}

\author{Yuma Tamura}
{Graduate School of Information Sciences, Tohoku University, Sendai, Japan}
{tamura@tohoku.ac.jp}
{0009-0001-5479-7006} 
{}

\author{Xiao Zhou}
{Graduate School of Information Sciences, Tohoku University, Sendai, Japan}
{zhou@tohoku.ac.jp}
{https://orcid.org/0000-0001-5473-7663} 
{}

\authorrunning{J. Cardinal et al.}

\Copyright{Jean Cardinal, Kevin Mann, Akira Suzuki, Takahiro Suzuki, Yuma Tamura and Xiao Zhou} 

\ccsdesc[500]{Theory of computation~Problems, reductions and completeness}
\ccsdesc[500]{Theory of computation~Parameterized complexity and exact algorithms}
\keywords{combinatorial reconfiguration, independent set, combinatorial shortest path, NP-completeness} 

\category{} 

\relatedversion{} 

\nolinenumbers
\EventEditors{Michal Kouck\'{y} and Daniela Petri\c{s}an}
\EventNoEds{2}
\EventLongTitle{51st International Symposium on Mathematical Foundations of Computer Science (MFCS 2026)}
\EventShortTitle{MFCS 2026}
\EventAcronym{MFCS}
\EventYear{2026}
\EventDate{August 24--28, 2026}
\EventLocation{Paris, France}
\EventLogo{}
\SeriesVolume{386}
\ArticleNo{32}
\funding{This work was partially supported by JSPS KAKENHI Grant Numbers {JP25K14980} and JP25K21148.}
\hideLIPIcs

\newcommand{\indG}[1]{G[#1]}

\newcommand{\intset}[1]{[#1]}
\newcommand{\intsets}[2]{[#1,#2]}
\newcommand{\numcomp}{\mathrm{cc}}

\newcommand{\flipC}{L_{C}}
\newcommand{\flipr}{L_{r}}
\crefname{paragraph}{Case}{Cases}
\newcommand{\prb}[1]{\textnormal{\scshape #1}}
\newcommand{\TRUE}{T}
\newcommand{\FALSE}{F}
\newcommand{\rev}[1]{#1}
\newcommand{\revrev}[1]{#1}

\begin{document}

\maketitle

\begin{abstract}
We initiate the study of the shortest reconfiguration problem for independent sets under the adjacency relation derived from the independent set polytope.
Given a graph and two independent sets, the problem asks for a shortest sequence transforming one into the other such that the subgraph induced by the symmetric difference of any two consecutive sets is connected.
This is equivalent to finding a shortest path on the $1$-skeleton of the independent set polytope.
We prove that the problem is \NP-hard even on planar graphs of bounded degree, as well as on split graphs.
Notably, the hardness for planar graphs of bounded degree still holds even when deciding whether the target can be reached in at most two steps.
For split graphs, we further show the \W[2]-hardness when parameterized by the number of steps, as well as the inapproximability of the optimal length.
As a consequence, we prove that the length of a shortest path between two vertices of a 0/1 polytope in $\mathbb{R}^n$ described by $O(n)$ linear inequalities is hard to approximate within a factor of $(1-\epsilon)\ln n$ for any constant $\epsilon >0$, unless $\P=\NP$.
On the positive side, we provide polynomial-time algorithms for block graphs, cographs, and bipartite chain graphs.
Moreover, for paths and cycles, we show that the optimal length of the shortest reconfiguration sequence exactly matches a trivial upper bound.
\end{abstract}


\section{Introduction}
\emph{Combinatorial reconfiguration}~\cite{reconf/ItoDHPSUU11} studies the structure of the solution space of a combinatorial problem by introducing an adjacency relation (called \emph{reconfiguration rule}) on the feasible solutions.
The resulting graph, called the \emph{reconfiguration graph}, has feasible solutions as vertices, with edges joining pairs of solutions related by the rule.
This framework provides a useful perspective on discrete solution spaces
and has been connected to several algorithmic tasks,
including optimization, counting, enumeration, and sampling.
In general, the number of feasible solutions may be exponential in the input size,
and hence the reconfiguration graph can be exponentially large.
This gives rise to a variety of algorithmic questions concerning the structure of
the reconfiguration graph.
We focus on two of them: \emph{reachability} and
\emph{shortest reconfiguration}.
Given two feasible solutions, the former asks whether there is a path (called a \emph{reconfiguration sequence}) between them in the reconfiguration graph, while the latter asks for the shortest path between them.

Among the various reconfiguration problems studied so far, \emph{independent set reconfiguration} (ISR) is one of the central topics.
In this problem, each feasible solution is an independent set of a graph.
The most widely studied rules include token jumping (TJ), token sliding (TS), and token addition and removal (TAR), and these models have led to a substantial body of work on algorithms and computational complexity (see \cite{reconfsurvey/BousquetMNS24,reconfsurvey/Heuvel13} for a comprehensive review).
We describe the related reconfiguration rules in \Cref{relatedwork}.

In all these models, an independent set is viewed as a placement of tokens on vertices, and a reconfiguration step modifies the solution by moving, adding, or removing a single token.
More recently, several extensions of these standard rules have also been considered, such as $k$-TJ~\cite{kTJ/SugaSTZ25} and $(k,d)$-TJ~\cite{k1TJ/KristanS25}, which allow one to modify multiple tokens in a single reconfiguration step.
These developments suggest that it is natural to study reconfiguration rules beyond the classical single-token settings.
Moreover, the latter is \rev{partly} motivated by a setting from multi-agent systems (MAS)~\cite{MAS/SternSFK0WLA0KB19}, showing that motivations from other areas can also lead to natural reconfiguration rules.

In this paper, we study a reconfiguration rule for independent sets that is induced by the graph of the independent set polytope.
For a graph $G=(V,E)$, the characteristic vector of an independent set $I$ is the vector $\bm{x}\in\{0,1\}^V$ such that $x_v=1$ if and only if $v\in I$.
The \emph{independent set polytope} of $G$ is the convex hull of the characteristic vectors of the independent sets of $G$~\cite{ISpolytope/CHVATAL1975138}.
A theorem of Chv\'atal characterizes adjacency in this polytope as follows.

\begin{theorem}[Chv\'atal~\cite{ISpolytope/CHVATAL1975138}]
Let $I$ and $J$ be two distinct independent sets of a graph $G$.
Then the characteristic vectors of $I$ and $J$ are adjacent in the independent set polytope of $G$ if and only if the induced subgraph $G[I\bigtriangleup J]$ is connected.
\end{theorem}

Motivated by this characterization, we consider the reconfiguration rule in which two independent sets $I$ and $J$ are adjacent whenever $G[I\bigtriangleup J]$ is connected.
\rev{
Note that this rule imitates TAR when $G[I\bigtriangleup J]$ is a single vertex, and TS when $G[I\bigtriangleup J]$ consists of two vertices with an edge.}
Under this rule, the resulting reconfiguration graph coincides with the graph of the independent set polytope.
Hence, our problem can be viewed as the shortest path problem on that polytope.
Shortest path problems on combinatorial polytopes have been studied for various structures, including associahedra~\cite{PMR/ItoKKKO22}, perfect matching polytopes~\cite{PMR/CardinalS25}, and polymatroids~\cite{polymatroids/CardinalS25}.
They are also closely related to the analysis of the simplex method for linear optimization~\cite{Pivotrule/LoeraKS22}.

\subsection{Related Work}\label{relatedwork}

\subsection*{Other Reconfiguration Rules for \prb{ISR}.}
As mentioned above, the \prb{ISR} problem is one of the most fundamental problems in the context of combinatorial reconfiguration.
Most previous work on \prb{ISR} has considered rules in which adjacency is defined by moving, adding, or removing tokens placed on vertices, e.g., TJ, TS, TAR.

For all three rules, it is known that deciding reachability is \PSPACE-complete even for planar subcubic graphs~\cite{reconf/HearnD05}, planar graphs with bounded bandwidth~\cite{ISR/Wrochna18}, and \W[1]-hard when parameterized by the number of tokens plus the length of the reconfiguration sequence~\cite{shortISR/BodlaenderGS21}.
Moreover, the reachability problem under TS is \PSPACE-complete on bipartite graphs~\cite{ISRbipartite/LokshtanovM19} and split graphs~\cite{TSsplit/BelmonteKLMOS21}, whereas the one under TJ is \NP-complete on bipartite graphs~\cite{ISRbipartite/LokshtanovM19} and polynomial-time solvable on split graphs.
On the positive side, there are many algorithmic results for the reachability problem, for example, linear-time solvability for trees under TJ and TS~\cite{ISRtree/DemaineDFHIOOUY15}.

More recently, several extensions beyond these standard rules have also been studied.
For instance, Suga et al.~\cite{kTJ/SugaSTZ25} introduced $k$-TJ, which allows up to $k$ tokens to jump simultaneously in a single reconfiguration step; Hatano et al.~\cite{hopISR/HatanoKIIM26} introduced $d$-jump, which allows a single token to move along a path of length $d$; and
Krist\'an and Svoboda~\cite{k1TJ/KristanS25} introduced $(k,d)$-TJ, which allows all tokens to move to vertices at distances of $d$.

\subsubsection*{Alternating Cycle Model for Perfect Matching Reconfiguration.}
A related reconfiguration problem, motivated by a combinatorial polytope, has been studied for perfect matchings, that is, edge sets such that every vertex is incident to exactly one in the set.
Ito et al.~\cite{PMR/ItoKKKO22} introduced a corresponding reconfiguration rule in which two perfect matchings are adjacent whenever their symmetric difference forms a single cycle.
Under this rule, deciding whether one perfect matching can be transformed into another in two steps is \NP-hard, whereas a polynomial-time algorithm is known for outerplanar graphs~\cite{PMR/ItoKKKO22}.
Cardinal and Steiner~\cite{PMR/CardinalS25} later showed that constant-factor approximation is \NP-hard even for subcubic graphs.
Moreover, assuming the exponential time hypothesis (ETH), they proved that approximating the length of a shortest path on the perfect matching polytope of an $n$-vertex graph within a factor of $\left(\frac{1}{4}-o(1)\right)\frac{\log n}{\log\log n}$ is hard.

\subsubsection*{From the Viewpoint of Linear Programming.}\label{relatedworkLP}
This line of research is also closely related to the complexity of the simplex method and pivot rules in linear optimization.
Since the seminal exponential lower bound of Klee--Minty~\cite{Pivotrule/KM72}, many pivot rules have been shown to have exponential or superpolynomial worst-case behavior.
In this context, it is natural to consider paths along which the objective value improves at every step.
Formally, a path on a polytope $P$ is called (strictly) \emph{monotone} with respect to a vector $c$ if the corresponding linear objective values increase monotonically along the path.
Such paths are directly relevant to pivot rules, since each step along a monotone path corresponds to an improving pivot.
De Loera, Kafer, and Sanit\`a showed that finding a shortest monotone path to an optimum is \NP-hard~\cite{Pivotrule/LoeraKS22}.
We give an interpretation of our result in this context in \Cref{polytopeinterpretation}.

A recent line of research has focused on restricting the structure of the underlying polytope.
In particular, Cardinal and Steiner established strong inapproximability results for perfect matching polytopes of bipartite graphs, which admit descriptions with a linear number of inequalities.
Black and Steiner later showed that computing a shortest monotone path to an optimum remains \NP-hard even on simple polytopes~\cite{polytopes/BlackS2026}.
On the other hand, Cunha et al.~\cite{polytopes/cunha2026} proved fixed-parameter tractability for graph associahedra parameterized by the distance, while establishing \W[2]-hardness and logarithmic-factor inapproximability for hypergraphic polytopes.

\subsection{Our Contribution}

We initiate the study of \textsc{SymISR}, the problem of finding a \emph{shortest} reconfiguration sequence under the rule that two independent sets $I$ and $J$ are adjacent whenever $G[I\bigtriangleup J]$ is connected.
The shortest variant is the natural algorithmic question in our setting, since reachability is trivial: any two independent sets of an $n$-vertex graph can be transformed into each other in $O(n)$ steps; see \Cref{sec:prbdef}.
We first show that it is \NP-complete to decide whether the shortest reconfiguration sequence between two given independent sets has length at most~$2$, even for planar graphs of bounded degree and degeneracy~$2$.
Since the case for length at most~$1$ can be verified in linear time, our result establishes a sharp complexity boundary with respect to the optimal length~$k$ of a reconfiguration sequence.
Moreover, unless $\P=\NP$, it implies that there is no polynomial-time $(3/2-\epsilon)$-approximation algorithm for any $\epsilon > 0$ (hence it admits no PTAS), and no {\XP} algorithm with respect to $k$, for such a restricted condition.

We therefore turn to graph classes with simpler structure.
Because \prb{SymISR} is defined in polyhedral terms, it is natural to focus on classes whose independent set polytopes admit well-structured descriptions.
In particular, Chv\'atal~\cite{ISpolytope/CHVATAL1975138} showed that, for perfect graphs, the independent set polytope is described by: \emph{nonnegativity}: $x_v\ge 0$ for every $v\in V$, and \emph{clique inequalities}: $\sum_{v\in C}x_v\le 1$ for every maximal clique $C$ of $G$.
We therefore focus on subclasses of perfect graphs.

Despite the existence of such well-structured polytopes, we prove that \prb{SymISR} is \NP-hard on split graphs, a proper subclass of perfect graphs.
Our reduction further implies \W[2]-hardness parameterized by the length of the reconfiguration sequence, and inapproximability for computing a shortest reconfiguration sequence.
On the algorithmic side, we give polynomial-time algorithms for block graphs, bipartite chain graphs, and cographs, all of which are proper subclasses of perfect graphs.
Note that forests, which are known to be graphs with degeneracy $1$, are a proper subclass of block graphs.
Combined with the hardness for graphs of degeneracy~$2$ and tractability for forests, this gives another tight boundary for \prb{SymISR} in terms of the degeneracy of the graph.
\Cref{fig:contribution} illustrates our contributions regarding perfect graphs.

\begin{figure}
    \centering
    \begin{tikzpicture}[
  every node/.style={
    rectangle,
    draw=black,
    fill=none,
    inner sep=2pt,
    outer sep=0pt,
    minimum height=1.8em
  }
]
    \node (P) at (-0.5,3.5) {Perfect};
    \node (C) at (-4,2) {Chordal};
    \node (DH) at (-0.5,2) {Distance-hereditary};
    \node (B) at (3,2) {Bipartite};
    \node (Sp) at (-5,0.5) {Split};
    \node (Bl) at (-2.2,0.6) {Block};
    \node (Co) at (-0.5,0.8) {Cographs};
    \node (Ch) at (3,0.5) {Bipartite chain};
    \node (T) at (-0.7,-0.5) {Tree};

  \draw (P.south) -- (C.north);
  \draw (P.south) -- (DH.north);
  \draw (P.south) -- (B.north);
  \draw (C.south) -- (Sp.north);
  \draw (C.south) -- (Bl.north);
  \draw (DH.south) -- (Bl.north);
  \draw (DH.south) -- (Co.north);
  \draw (B.south) -- (Ch.north);
  \draw (Bl.south) -- (T.north);
  \draw (B.south) -- (T.north);
  \draw (DH.south) -- (Ch.north);

  \draw[blue, very thick] (5,3) .. controls (-5.5, 3.5) and (-3, 0) .. (-5.5, -0.5);
  \draw[red, very thick] (5,1.3) .. controls (-6, 2) and (-3.5, -0.5) .. (-5.5, -0.6); 
  \node[blue, draw = none] (NPh) at (4.2, 3.3) {\NP-hard};
  \node[red, draw = none] (poly) at (4.5,1.0) {in \P};
  \node[gray, draw = none] (open) at (4.5, 2) {Open};
  
\end{tikzpicture}
    \caption{A selection of our contributions focusing on subclasses of perfect graphs.
    In addition, the paper establishes the \NP-hardness of planar graphs with bounded degree and degeneracy, and provides an exact formula for the length of sequences (and hence linear-time algorithms) in paths and cycles.
    The complexity of \prb{SymISR} for bipartite graphs and distance-hereditary graphs remains open; we leave this as a future direction.}
    \label{fig:contribution}
\end{figure}
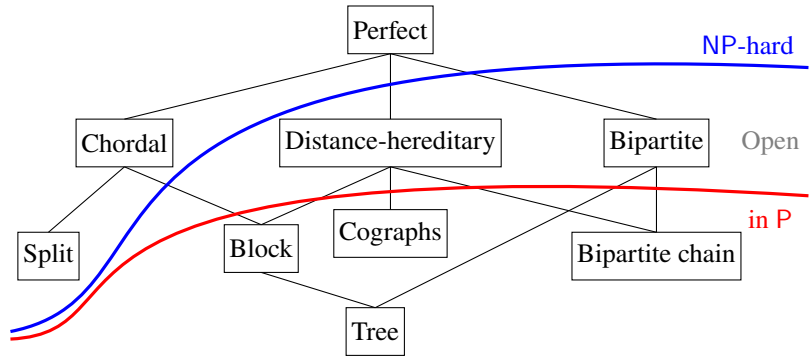

It is also natural to ask when the trivial upper bound (described in \Cref{sec:preliminaries}) derived by the definition of adjacency is tight.
We show that this is the case for paths and cycles, and hence obtain linear-time algorithms for them.


\subsubsection*{Consequences for Shortest Paths on Polytopes}
\label{polytopeinterpretation}

The inapproximability result for split graphs can be translated into a statement about shortest paths on independent set polytopes.
Indeed, since split graphs are perfect and have only \(O(n)\) maximal cliques, the independent set polytope of an \(n\)-vertex split graph is described by \(O(n)\) clique inequalities, together with the $O(n)$ nonnegativity constraints~\cite{ISpolytope/CHVATAL1975138}.
Hence, we directly obtain that the shortest path problem is hard to approximate on 0/1 polytopes given by linear inequalities.

\begin{corollary}\label{inapprox}
Given a polytope $P \subseteq \mathbb{R}^n$ described by $O(n)$ inequalities
with $0/1$ coefficients and two vertices $u,v$ of $P$, the length of a shortest  path from $u$ to $v$ in the skeleton of $P$ cannot be approximated within a factor of $(1-\varepsilon)\ln n$ in polynomial time for any fixed constant $\varepsilon >0$, unless $\P=\NP$.
\end{corollary}

Moreover, the same reduction can be equipped with a suitable linear objective function to show that the relevant paths are monotone.
Hence, our result implies the following (see the bottom of \Cref{sec:split} for proof).

\begin{corollary}\label{monotoneinapprox}
Given a polytope $P \subseteq \mathbb{R}^n$ described by $O(n)$ inequalities
with $0/1$ coefficients, a vertex $u$ of $P$, and a linear objective function
$c:\mathbb{R}^n \to \mathbb{R}$, the length of a shortest monotone path from
$u$ to an optimal solution of the linear program defined by $P$ and $c$
cannot be approximated within a factor of $(1-\varepsilon)\ln n$ in polynomial
time for any fixed constant $\varepsilon > 0$, unless $\P=\NP$.
\end{corollary}

This improves on the lower bound of Cardinal and Steiner under the ETH: we obtain a logarithmic inapproximability result under the weaker assumption $\P\ne \NP$.
A related result was recently obtained for hypergraphic polytopes by Cunha et al.~\cite{polytopes/cunha2026}.
Hypergraphic polytopes, however, can have exponentially many facets, hence cannot be input explicitly by linear inequalities.
In contrast, our lower bound holds for polytopes with only $O(n)$ facets.


As a consequence of Corollary~\ref{monotoneinapprox}, it is unlikely that any pivot rule can simultaneously run in polynomial time and guarantee a monotone path whose length is within a logarithmic factor of the optimum, when a given polytope has such a structured formulation.
\rev{On the positive side, an $O(n)$-approximation algorithm is known when $P \subseteq \mathbb{R}^n$ is $0/1$-polytope, namely the convex hull of a set of $0/1$ vectors.
Indeed, Naddef showed that every $0/1$-polytope has diameter at most $n$~\cite{polytope:Naddef89}, and Black et al. further showed that there is a pivot rule for the simplex method to find a monotone path of length at most $n$~\cite{polytope:BlackLKS25}.
Thus, our results indicate that improving this $O(n)$-factor approximation to a logarithmic factor is unlikely even for polytopes described by only linearly many inequalities.
}

In what follows, we introduce preliminaries, show two hardness results and their consequences, and then present our algorithmic results, along with results on when the trivial upper bound is tight.
\rev{Due to the space limitation, proofs and discussions marked $(\ast)$ are omitted in this paper.}

\section{Preliminaries}\label{sec:preliminaries}

Let $G = (V,E)$ be a graph. We also denote by $V(G)$ and $E(G)$ the vertex set and the edge set of $G$, respectively. 
All the graphs considered in this paper are finite, simple, and undirected.
We treat $n$ as the number of vertices and $m$ as the number of edges, unless specifically stated otherwise.
For a vertex $v$ of $G$, the \emph{open neighborhood} $N(v)$ and the \emph{degree} $d(v)$ of $v$ are defined as $N(v) = \{ w \in V \mid vw \in E \}$ and $d(v) = |N(v)|$, respectively.
We also define the \emph{closed neighborhood} $N[v] = N(v) \cup \{v\}$.
For a vertex subset $S \subseteq V(G)$, we denote by $\indG{S}$ the subgraph induced by $S$.
A \emph{connected component} of a graph is a connected induced subgraph that is not contained in another connected subgraph.
For a vertex subset $V'\subseteq V$, let $\numcomp(V')$ be the number of connected components in a graph $G[V']$.
For two disjoint graphs $G_1 = (V_1, E_1)$ and $G_2 = (V_2, E_2)$, their \emph{disjoint union} is defined as $G_1 \oplus G_2 = (V_1 \cup V_2, E_1 \cup E_2)$, and their \emph{join} is defined as $G_1 \otimes G_2 = (V_1 \cup V_2, E_1 \cup E_2 \cup \{v_1v_2 \mid v_1 \in V_1, v_2 \in V_2\})$.
For two positive integers $i$ and $j$ such that $i \le j$, we denote $\intsets{i}{j} = \{i,i+1,\dots, j\}$.
We use the shorthand $\intset{j} = \intsets{1}{j}$.

An \emph{independent set} $I$ of $G$ is a vertex subset of $G$ such that there are no edges between any pair of vertices in $I$.
A \emph{clique} $C$ of $G$ is a vertex subset of $G$ such that there is an edge between $u$ and $v$ for every pair of distinct vertices $u,v\in C$.
A \emph{degeneracy} $d$ of $G$ is the minimum number such that every subgraph has at least one vertex of degree at most $d$.
For the basic concepts and definitions in parameterized complexity theory, see some textbooks, e.g.,~\cite{CyganFKLMPPS15}.

\subsubsection*{Problem.}\label{sec:prbdef}
Two independent sets $I$ and $J$ of a graph $G$ are said to be \emph{adjacent}
if $G[I \bigtriangleup J]$ is connected.
In this case, $J$ is obtained from $I$ by a \emph{flip} on a connected vertex set $X = I \triangle J$, where a flip is the operation of reversing the membership of each vertex in $X$.
\rev{
Note that, for an independent set $I$, the vertex set obtained by flipping a connected vertex subset is not always independent; in this paper, we consider only those for which the resulting vertex set is independent.
}

Given a graph $G$, two independent sets $I_s$ and $I_t$, and an integer $k$, the \prb{SymISR} problem asks whether there exists a reconfiguration sequence
$\sigma=\langle I_s=I_0,I_1,\dots,I_\ell=I_t\rangle$ of length $\ell\le k$ \rev{such that for every $i \in [\ell]$, $I_i$} is an independent set of $G$ and every two consecutive independent sets are adjacent.
Equivalently, $G[I_{i-1}\bigtriangleup I_i]$ is connected for every $i\in[\ell]$.

We begin with some basic observations.
First, the following observation is immediate, as the connectivity of $G[I_s \triangle I_t]$ can be verified in linear time.

\begin{observation}\label{obs:length1}
    When $k=1$, \prb{SymISR} can be solved in linear time for general graphs.
\end{observation}

Second, one can transform $I_s$ into $I_t$ by flipping the connected components of $G[I_s \bigtriangleup I_t]$ one by one.
Thus, we refer to $\numcomp(I_s \triangle I_t)$ as a \emph{trivial upper bound}.
We remark that the following is also the certificate for membership in \NP. 

\begin{observation}
    For a given graph $G$ and two independent sets $I_s$ and $I_t$, a shortest reconfiguration sequence has length at most $\numcomp(I_s\bigtriangleup I_t)$.
\end{observation}
\begin{proof}
\rev{
Let $C_1, C_2,...,C_i$ be connected components of $G[I_s\bigtriangleup I_t]$, where $i=\numcomp(I_s\bigtriangleup I_t)$.
Consider the sequence $\sigma = \langle I_s = I_0, I_1, ..., I_i = I_t \rangle$ of vertex sets, where $I_j$ is obtained from $I_{j-1}$ by a flip on $C_j$ for $j\in [i]$.
We claim that $\sigma$ is a reconfiguration sequence from $I_s$ to $I_t$ by showing that $I_j$ is independent for every $j\in [0,i]$.
}

\rev{For the sake of contradiction, let $j$ be the minimum integer such that $I_j$ is not independent. Clearly, $j > 0$. Then, there exist vertices $v, v' \in I_j$ such that $v v' \in E(G)$. Moreover, since $I_{j-1}$ is independent due to the minimality of $j$, we have $v \in V(C_j)$ or $v' \in V(C_j)$.
If both $v$ and $v'$ belong to $V(C_j)$, we obtain $v, v' \in I_t$ because a flip is performed at most once on each vertex, which contradicts the fact that $I_t$ is independent.
Suppose next that $v \in V(C_j)$ and $v' \notin V(C_j)$. Note that the case where $v \notin V(C_j)$ and $v' \in V(C_j)$ is symmetric. Since a flip is performed at most once on each vertex, we have $v \in I_t$. Moreover, no flip is performed on $v'$; otherwise, $v$ and $v'$ would be contained in the same connected component $C_j$, since $v v' \in E(G)$. Consequently, we have $v' \in I_t$, contradicting the independence of $I_t$.}
\end{proof}



\section{Hardness for Planar Graphs}\label{sec:hard}

A graph $G$ is said to be \emph{planar} if $G$ can be drawn in the plane without any edge crossing.
We begin with proving the \NP-completeness of \prb{SymISR} on planar graphs, even under very restrictive conditions.
Specifically,  we show the following theorem.
\begin{theorem}\label{the:planarhard}
    \prb{SymISR} is \NP-complete for planar graphs with maximum degree at most 6 and degeneracy 2, even when $k=2$.
\end{theorem}


We reduce \textsc{Linked Planar SAT} (\textsc{LP-SAT}) to \textsc{SymISR}. 
First, we define \textsc{LP-SAT}.
Let $\phi$ be a CNF formula with clause set $C=\{c_1,c_2,\dots,c_m\}$ and variable set $X=\{x_1,x_2,\dots,x_n\}$.
The \emph{incidence graph} of $\phi$, denoted by $G(\phi)$, is the graph with vertex set $V(G(\phi))=X\cup C$ and edge set
$ E(G(\phi))=\{x_ic_j \mid x_i \text{ occurs in } c_j\}$.
A \emph{link} is an edge set $ \{x_ix_{i+1}\mid i\in [1,n-1]\}\cup \{c_ic_{i+1}\mid i\in [1,m-1]\}\cup \{x_nc_1,c_mx_1\}$.
Moreover, we refer to the graph obtained by adding a link to $G(\phi)$ as the \emph{linked incidence graph} of $\phi$.
\textsc{LP-SAT} is the problem of deciding whether $\phi$ is satisfiable, given a CNF formula $\phi$ such that the linked incidence graph of $\phi$ is planar.
We assume that such a planar embedding is given.
It is known that \textsc{LP-SAT} remains NP-complete even under the following restrictions~\cite{planarSAT/Pilz19}:
\begin{enumerate}[(1)]
    \item each clause contains at most three variables;
    \item each variable occurs in at most three clauses; and
    \item in a planar embedding, the edges corresponding to positive occurrences lie in the interior of the link of $\phi$, whereas those corresponding to negative occurrences lie in the exterior.
\end{enumerate}
The example of such an input and corresponding linked incidence graph is in \Cref{fig:planar_construction}.
We assume that the input formula $\phi$ satisfies the conditions above.
Given a CNF formula $\phi$ together with a planar embedding of its linked incidence graph, we now construct an instance $I=(G',I_s,I_t,k)$ of \textsc{SymISR}.

\subsection{Construction}
We describe the construction in three parts: the variable gadget, the clause gadget, and the occurrence paths connecting them.
First, we construct the \textit{variable gadget}, whose vertex set consists of $X'=X_V\cup Y$, where $X_V=\{x_{iT},x_{iF}\mid i\in [n]\}$ and $Y=\{y_i\mid i\in [0,n]\}$.
The edge set of the variable gadget is $\{y_{i-1}x_{iT},y_{i-1}x_{iF},x_{iT}y_i,x_{iF}y_i,x_{iT}x_{iF}\mid i\in [n]\}$.
Next, the \textit{clause gadget} consists of $m$ disjoint copies of $K_2$, that is, the vertex set $\{c_{j},c_{j}'\mid j\in [m]\}$ and the edge set $\{c_{j}c_{j}'\mid j\in [m]\}$.
Finally, we connect the variable gadget and the clause gadget by occurrence paths.
For every positive occurrence (resp.\ negative occurrence) of $x_i\in X$ in $c_j\in C$, we connect $x_{iT}$ (resp. $x_{iF}$) and $c_{j}$ by a path $(x_{iT}, a_{ij}, b_{ij}, c_{j})$ (resp. $(x_{iF}, a_{ij}, b_{ij}, c_{j})$), completing the construction of $G'$.
We write $A=\{a_{ij}\mid \text{one of }x_i \text{ and } \neg x_i \text{ occurs in } c_j\}$ and $B=\{b_{ij}\mid \text{one of }x_i \text{ and } \neg x_i \text{ occurs in } c_j\}$.
We set $I_s =Y \cup A\cup\{c_{j}'\mid j\in [m]\}$, $I_t=Y \cup A\cup\{c_{j}\mid j\in [m]\}$, and $k=2$.
Immediately, both $I_s$ and $I_t$ are independent sets.

\begin{claim}\label{clm:planar}
    $G'$ is a planar graph with maximum degree $6$, and degeneracy $2$.
\end{claim}
\iftrue
\begin{proof}
\revrev{To see this, start with a planar embedding of the linked incidence graph.
For each $i \in [n-1]$, place $y_i$ on the link edge between $x_i$ and
$x_{i+1}$.
Place $y_0$ (resp.\ $y_n$) on the link sufficiently close to $x_1$
(resp.\ $x_n$) on the side opposite to $y_1$ (resp.\ $y_{n-1}$).
For each $i \in [n]$, embed the variable gadget by placing the edge
$x_{iT}x_{iF}$ in a small neighborhood of $x_i$ with $x_{iT}$ on the
inner side and $x_{iF}$ on the outer side of the link, and adding the
edges from $\{x_{iT}, x_{iF}\}$ to $\{y_{i-1}, y_i\}$.
Finally, replace the edges of $G(\phi)$ with the occurrence paths.
Recall that all edges incident to $x_i$ that correspond to positive literal occurrences lie on one side of the link, whereas all edges incident to $x_i$ that correspond to negative literal occurrences lie on the other side.
Hence, the occurrence paths corresponding to positive literals can be routed through $x_{iT}$, and those corresponding to negative literals can be routed through $x_{iF}$, without creating any crossings.
See~\Cref{fig:planar_construction} for an illustration.
This transformation preserves planarity, and thus $G'$ is planar.}

\revrev{Moreover, observe that the degree of each vertex of $G'$ is at most $6$ for vertices in $X_V$, at most $4$ for vertices in $Y\cup\{c_{j},c_{j}'\mid j\in [m]\}$, and at most $2$ for vertices in $A\cup B$.
Additionally, we can verify that the degeneracy of $G'$ is $2$.}
\end{proof}
\fi

\begin{figure}[t]
    \centering
    \begin{tikzpicture}
        \node[anchor=south west, inner sep=0] (img) at (0,0)
          {\includegraphics[width=0.9\textwidth]{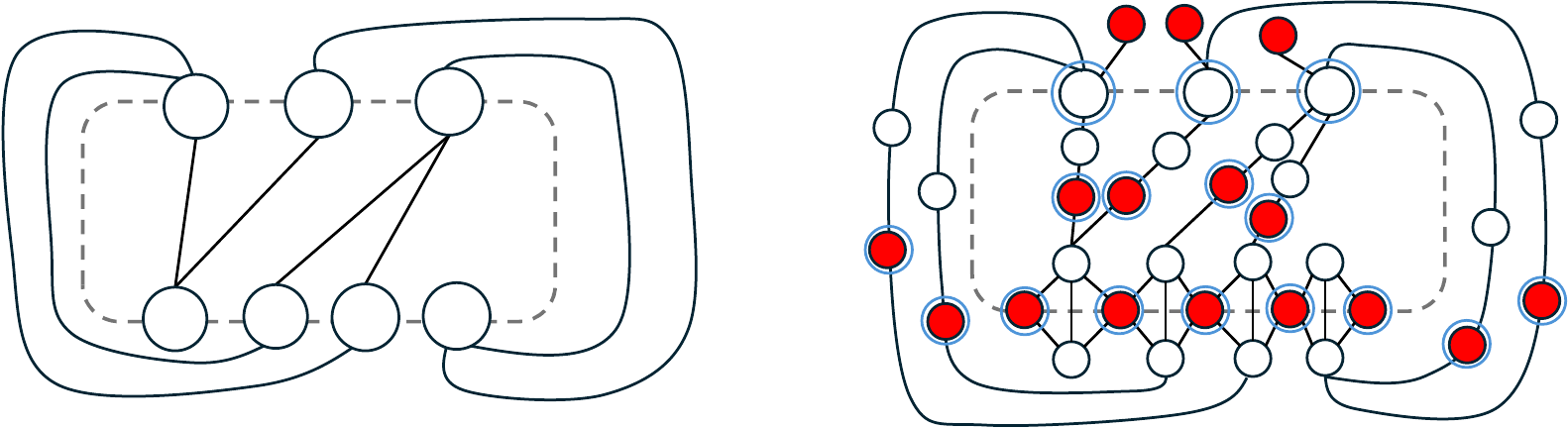}};

        {\small
        \begin{scope}[x={(img.south east)}, y={(img.north west)}]
            \node at (0.11,0.24) {$x_1$};
            \node at (0.175,0.24) {$x_2$};
            \node at (0.232,0.24) {$x_3$};
            \node at (0.293,0.24) {$x_4$};
            \node at (0.125,0.75) {$c_1$};
            \node at (0.205,0.75) {$c_2$};
            \node at (0.285,0.75) {$c_3$};
        \end{scope}
        }
    \end{tikzpicture}
    \caption{An illustration of the construction, where $\phi= (x_1\vee \neg x_2 \vee \neg x_3)\wedge (x_1\vee \neg x_4)\wedge (x_2\vee x_3\vee \neg x_4)$. (\emph{Left}.) A linked incidence graph of $\phi$. A link is drawn by a dashed line. (\emph{Right.}) A corresponding instance $I=(G',I_s,I_t,2)$ of \prb{SymISR}. Red vertices indicate $I_s$, and and blue circles indicate $I_t$. A variable gadget is arranged along the link, with $x_{iT}$ placed on the inner side of the link and $x_{iF}$ on the outer side for each $i \in [n]$.}
    \label{fig:planar_construction}
\end{figure}
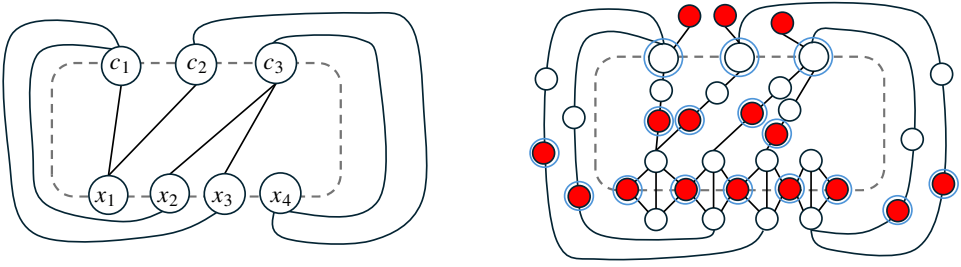

\subsection{Correctness}

Since $k=2$, it suffices to show that $\phi$ is satisfiable if and only if there exists an intermediate independent set $I_1$ that is adjacent to both $I_s$ and $I_t$.

($\Rightarrow$) Assume that there is a satisfying assignment $\boldsymbol{x}\in \{\TRUE, \FALSE\}^n$ for $\phi$, which is a vector such that $i$-th component corresponds to the value of $x_i$.
We define $
X_1 =\{x_{iT}\mid x_i = \TRUE \text{ in } \boldsymbol{x}\}\cup \{x_{iF}\mid x_i = \FALSE \text{ in } \boldsymbol{x}\}
$.
For each occurrence of $x_i$ in $c_j$, let $\ell_{ij}$ denote the literal of $c_j$ corresponding to that occurrence.
We include $b_{ij}$ in $I_1$ if $\ell_{ij}$ is evaluated as true under $\boldsymbol{x}$, and include $a_{ij}$ otherwise.
Thus,
$
I_1 = X_1 \cup \{c_{j}'\mid j\in [m]\}\cup \{z_{ij}\mid x_i \text{ occurs in } c_j\},
$
where $z_{ij}=b_{ij}$ if $\ell_{ij}$ is evaluated as true under $\boldsymbol{x}$, and $z_{ij}=a_{ij}$ otherwise.
The following proposition verifies that $I_1$ is a proper intermediate set.
\begin{proposition}\label{lem:planarright}
    $I_1$ is an independent set of $G$, and both $G[I_s\bigtriangleup I_1]$ and $G[I_1\bigtriangleup I_t]$ are connected.
\end{proposition}
\iftrue
\begin{proof}
\revrev{We first claim that $I_1$ is an independent set.
By construction, clearly, $X_1 \cup \{c_{j}'\mid j\in [m]\}$ is an independent set.
It suffices to show that $z_{ij}$ has no neighbors in $I_1$.
If $z_{i,j} = b_{i,j}$, then the claim follows from the fact that $c_{j}\notin I_1$ and exactly one of $a_{ij}$ and $b_{ij}$ belongs to $I_1$.
If $z_{i,j} = a_{i,j}$, then the literal $\ell_{ij}$ is evaluated as false under $\boldsymbol{x}$.
Therefore, the vertex corresponding to $\ell_{ij}$, which is adjacent to $a_{ij}$, does not belong to $I_1$.
We conclude that $I_1$ is an independent set.}

\revrev{The remaining task is to show that both $I_s\bigtriangleup I_1$ and $I_1\bigtriangleup I_t$ are connected.
First, consider $I_s\bigtriangleup I_1$.
Observe that $
I_s\bigtriangleup I_1
=
Y \cup X_1
\cup \{a_{ij},b_{ij}\mid x_i \text{ occurs in } c_j \text{ and } z_{ij}=b_{ij}\}.
$
The set $Y\cup X_1$ induces a path from $y_0$ to $y_n$, since exactly one of $x_{iT}$ and $x_{iF}$ belongs to $I_1$ for each $i\in [n]$.}
\revrev{Consider two integers $i\in [n]$ and $j\in [m]$ such that $a_{ij},b_{ij}\in I_s\bigtriangleup I_1$.
By the construction of $I_1$, since $z_{ij}=b_{ij}$, the vertex $x_{i\ast}$ with $\ast \in \{T,F\}$ corresponding to the literal $\ell_{ij}$ belongs to $I_1$.
The three vertices $a_{ij}$, $b_{ij}$, and $x_{i\ast}$ induce a path in $G[I_s\bigtriangleup I_1]$, which shares the vertex $x_{i\ast}$ with $G[Y\cup X_1]$.
Thus, $G[I_s\bigtriangleup I_1]$ is connected.}

\revrev{Next, consider $I_1\bigtriangleup I_t$.
Observe that $
I_1\bigtriangleup I_t
=
Y \cup X_1
\cup \{a_{ij},b_{ij}\mid x_i \text{ occurs in } c_j \text{ and } z_{ij}=b_{ij}\}
\cup \{c_{j},c_{j}'\mid j\in [m]\}
$.
As claimed above, $Y \cup X_1
\cup \{a_{ij},b_{ij}\mid x_i \text{ occurs in } c_j \text{ and } z_{ij}=b_{ij}\}$ induces a connected subgraph.
Now let $j\in [m]$.
Since $\boldsymbol{x}$ is a satisfying assignment of $\phi$, there exists $i\in [n]$ such that $x_i$ occurs in $c_j$ and $\ell_{ij}$ is evaluated as true under $\boldsymbol{x}$.
For this choice of $i$, we have $z_{ij}=b_{ij}$ by the definition of $I_1$.
Hence, the vertices $b_{ij}$, $c_{j}$, $c_{j}'$ induce a path in $G[I_1\bigtriangleup I_t]$.
Therefore, $G[I_1\bigtriangleup I_t]$ is connected.}
\end{proof}
\fi

($\Leftarrow$) Assume that there exists an independent set $I_1$ adjacent to both $I_s$ and $I_t$.
We can show that there exists a truth assignment $\bm{x}\in \{\TRUE, \FALSE\}^n$ corresponding to $I_1$, assuming the existence of $I_1$.
The core idea of our proof is stated as follows.

\begin{lemma}\label{lem:clause}
For every $j\in [m]$, we have $c_{j}'\in I_1$.
\end{lemma}

\iftrue
\begin{proof}
\revrev{For the sake of contradiction, assume that there is an integer $j\in [m]$ such that $c_{j}'\notin I_1$.
Since $I_1$ is independent and $c_{j}c_{j}'\in E(G')$, at most one of $c_{j}$ and $c_{j}'$ belongs to $I_1$.}

\revrev{First, suppose that $c_{j}\notin I_1$.
Then $c_{j}\notin I_s\bigtriangleup I_1$, whereas $c_{j}'\in I_s\bigtriangleup I_1$ due to the assumption that $c_{j}'\notin I_1$ (and $c_j'\in I_s$ by construction).
Since $c_{j}'$ is adjacent only to $c_{j}$, the vertex $c_{j}'$ is an isolated vertex in $G[I_s\bigtriangleup I_1]$.
Due to the connectivity of $G[I_s\bigtriangleup I_1]$, we have $I_s\bigtriangleup I_1 = \{ c_{j}' \}$, implying $I_1 = I_s \setminus \{ c_{j}' \}$.
This yields $I_1\bigtriangleup I_t = \{ c_{j}, c_{j}' \mid j \in [m] \} \setminus \{ c_{j}' \}$.
However, $G[I_1\bigtriangleup I_t]$ is not connected by the construction of $G'$, a contradiction.}

\revrev{Next, suppose that $c_{j}\in I_1$.
Then both $c_{j}$ and $c_{j}'$ belong to $I_s\bigtriangleup I_1$.
Moreover, it follows that $b_{ij} \notin I_1 \cap N(c_{j})$ for $i \in [n]$ and $j \in [m]$.
Since $b_{ij} \notin I_s$, we have $b_{ij} \notin I_s\bigtriangleup I_1$, and hence $c_{j}$ and $c_{j}'$ induce a connected component of $G[I_s\bigtriangleup I_1]$.
Due to the connectivity of $G[I_s\bigtriangleup I_1]$, we have $I_s\bigtriangleup I_1 = \{ c_{j}, c_{j}' \}$, implying $I_1 = I_s \setminus \{ c_{j}, c_{j}' \}$.
By the same argument as above, we can derive a contradiction.}
\end{proof}

\revrev{By \Cref{lem:clause}, we have $c_{j}'\in I_1$ for every $j\in [m]$.
Since $I_1$ is independent and $c_{j}c_{j}'\in E(G')$, it follows that $c_{j}\notin I_1$ for every $j\in [m]$.
Thus, $\{c_{j},c_{j}' \mid j \in [m]\} \subseteq I_1\bigtriangleup I_t$.}

\revrev{Fix $j\in [m]$.
To make $G[I_1\bigtriangleup I_t]$ connected, $I_1\bigtriangleup I_t$ must contain $b_{ij}$ for some $i \in [n]$.
Since $b_{ij} \notin I_t$, we have $b_{ij} \in I_1$.
We also have $a_{ij}\notin I_1$ because $I_1$ is independent and $a_{ij}b_{ij}\in E(G')$.
Hence, $\{a_{ij}, b_{ij}\}\subseteq I_s\bigtriangleup I_1$, since $a_{ij}\in I_s$ and $b_{ij}\notin I_s$.
Let $x_{i\ast}$ be the unique vertex in $\{x_{iT},x_{iF}\}$ adjacent to $a_{ij}$.}

\revrev{Assume that $x_{i\ast}\notin I_1$.
Then $x_{i\ast}\notin I_s\bigtriangleup I_1$, since $x_{i\ast}\notin I_s$.
Furthermore, $c_{j}\notin I_s$ and $c_{j}\notin I_1$ imply that $c_{j}\notin I_s\bigtriangleup I_1$.
Therefore, in $G[I_s\bigtriangleup I_1]$, the set $\{a_{ij}, b_{ij}\}$ induces a connected component of $G[I_s\bigtriangleup I_1]$. 
This implies that $I_1 = I_s \cup \{ b_{ij} \} \setminus \{a_{ij}\}$.
However, one can verify that $G[I_1\bigtriangleup I_t]$ is not connected, a contradiction.
Thus, $x_{i\ast}\in I_1$.}

\revrev{We now define an assignment $\boldsymbol{x}$ by setting $x_i=\TRUE$ if $x_{iT}\in I_1$ and $x_i=\FALSE$ if $x_{iF}\in I_1$; if neither belongs to $I_1$, we assign $x_i$ arbitrarily.
This is well-defined because $I_1$ is independent and $x_{iT}x_{iF}\in E(G')$ for every $i\in [n]$.
Let $j\in [m]$.
By the argument above, there exists $i\in [n]$ such that $x_i$ occurs in $c_j$, $b_{ij}\in I_1$, and the vertex in $\{x_{iT},x_{iF}\}$ adjacent to $a_{ij}$ belongs to $I_1$.
Hence, the corresponding literal in $c_j$ is evaluated as true by $\boldsymbol{x}$.
Therefore, every clause of $\phi$ is satisfied by $\boldsymbol{x}$; consequently, $\phi$ is satisfiable.
This concludes the proof of \Cref{the:planarhard}.}
\fi

Our reduction further yields the following corollaries as consequences.
\begin{corollary}
    \prb{SymISR} cannot be approximated to within a factor smaller than $3/2$ in polynomial time unless $\P= \NP$, even when given graph $G$ is planar with maximum degree $6$ and degeneracy $2$.
\end{corollary}

\begin{corollary}
    \prb{SymISR} does not admit an {\XP} algorithm parameterized by $k$, even for planar graphs with maximum degree~$6$ and degeneracy~$2$, unless $\P=\NP$.
\end{corollary}

\section{On the Subclasses of Perfect Graphs}\label{sec:perfect}
\subsection{NP-completeness and inapproximability for split graphs}\label{sec:split}

A graph $G$ is \emph{split} if the vertex set of $G$ can be partitioned into exactly one clique and exactly one independent set.
This section shows the following theorem.

\begin{theorem} \label{thm:NPhard_split}
    \prb{SymISR} is \NP-hard for split graphs.
\end{theorem}

\Cref{thm:NPhard_split} is proved by reducing \textsc{Dominating Set}, which is well-known to be NP-hard, to \textsc{SymISR} on split graphs.
A vertex set $D$ of a graph $G$ is a \emph{dominating set} if $\bigcup_{v\in D} N[v] = V(G)$.
Given a graph $G$ and a non-negative integer $k$, the \textsc{Dominating Set} problem asks whether $G$ has a dominating set of size at most $k$. 
Since every dominating set contains an inclusion-wise minimal dominating set as a subset, any dominating set of size at most $k$ contains a minimal dominating set of size at most $k$. Therefore, asking whether $G$ has a dominating set of size at most $k$
is equivalent to asking whether it has an inclusion-wise minimal dominating set of size at most $k$.
\revrev{Moreover, there is a simple algorithm that runs in $O(|V(G)|^k)$ time for a fixed solution size $k$.}

\revrev{
We transform an instance $(G, k)$ of \textsc{Dominating Set} to an instance $(G', I_s, I_t, k')$ of \textsc{SymISR}.
The number of vertices of $G$ is denoted by $n$, and let $V(G) = \{ v_1, v_2, \dots, v_n \}$.
We define two copies $V_1$ and $V_2$ of $V(G)$.
Specifically, for a vertex $v_i \in V(G)$ for $i \in [n]$, the vertices $v_i^1 \in V_1$ and $v_i^2 \in V_2$ are defined.
We also introduce a vertex $v_0^1$.
The vertex set of $G'$ is represented as $V_1 \cup V_2 \cup \{v_0^1\}$.
Next, we define the edge set of $G'$.
If $v_j \in N[v_i]$ for integers $i,j \in [n]$, then $v_i^1$ and $v_j^2$ are joined by an edge.
Furthermore, an edge is introduced for every pair of distinct vertices in $V_1 \cup \{v_0^1\}$.
Finally, set $I_s = V_2 \cup \{v_0^1\}$, $I_t = \{v_0^1\}$, and $k' = k+1$.
This completes the construction of $(G', I_s, I_t, k')$.
An example of our construction is shown in \Cref{fig:split}.
}
\revrev{
Observe that the reduction runs in polynomial time, \revrev{and $G'$ is a split graph with $2|V(G)|+1$ vertices}, since $V^1\cup \{v^1_0\}$ is a clique of $n+1$ vertices and $V^2$ is an independent set of $n$ vertices. }

\revrev{In the following, we show the lemma stated below, which completes the reduction.
\begin{lemma}\label{lem:equivalence_dom_SymISR}
    $G$ has a dominating set of size at most $k$, if and only if there is a reconfiguration sequence of length at most $k'$ from $S$ to $T$ of $G'$.
\end{lemma}}

\begin{figure}
    \centering
    \begin{tikzpicture}
\begin{scope}[xshift = -1cm, yshift = -1.5cm]
\node[draw=black, circle, minimum size=6mm, inner sep=0pt] (v1) at (-0.5,1.2){$v_1$};
\node[draw=black, circle, minimum size=6mm, inner sep=0pt] (v2) at (-1.5,0.3){$v_2$};
\node[draw=black, circle, minimum size=6mm, inner sep=0pt] (v3) at (0.5,0.3){$v_3$};
\node[draw=black, circle, minimum size=6mm, inner sep=0pt] (v4) at (-1.5,-1){$v_4$};
\node[draw=black, circle, minimum size=6mm, inner sep=0pt] (v5) at (0.5,-1){$v_5$};
\draw (v1.south west) -- (v2.north east);
\draw (v1.south east) -- (v3.north west);
\draw (v2.south) -- (v4.north);
\draw (v2.east) -- (v3.west);
\draw (v4.east) -- (v5.west);
\node (G) at (-1, 2) {\Large $G$};
\node (k) at (0,2) {$k=2$};
    
\end{scope}
\begin{scope}[xshift = 1.2cm, yshift = -1.3cm]
\node (G) at (5, 2) {\Large $G'$};
\node (G) at (6, 2) {$k'=3$};

\node[draw, very thick, dotted, rectangle, rounded corners=7pt, minimum width=7.5cm, minimum height=2cm, label=left:{\Large $V^1\cup \{u\}$}] (S1) at (5,0.6) {};

\node[label=left:{\Large $V^2$}] (V^2) at (1.2,-1.3) {};

\node[draw=black, fill = red, fill opacity = 0.2, text opacity=1, circle, minimum size=6mm, inner sep=0pt] (z1) at (5,1.1){$v^1_0$};
\node[draw=blue, circle, minimum size=7.5mm, inner sep=0pt] (z2) at (5,1.1){};
\node[draw=black, circle, minimum size=6mm, inner sep=0pt] (x1) at (1.5+0.5,0.3){$v^1_1$};
\node[draw=black, circle, minimum size=6mm, inner sep=0pt] (x2) at (3+0.5,0.3){$v^1_2$};
\node[draw=black, circle, minimum size=6mm, inner sep=0pt] (x3) at (4.5+0.5,0.3){$v^1_3$};
\node[draw=black, circle, minimum size=6mm, inner sep=0pt] (x4) at (6+0.5,0.3){$v^1_4$};
\node[draw=black, circle, minimum size=6mm, inner sep=0pt] (x5) at (7.5+0.5,0.3){$v^1_5$};

\begin{scope}[xshift = 1.5cm]

\node[draw=black, fill = red, fill opacity = 0.2, text opacity=1, circle, minimum size=6mm, inner sep=0pt] (y11) at (0,-1.4){$v^2_1$};

\begin{scope}[xshift = 1.7cm]
    \node[draw=black, fill = red, fill opacity = 0.2, text opacity=1, circle, minimum size=6mm, inner sep=0pt] (y21) at (0,-1.4){$v^2_2$};
\end{scope}

\begin{scope}[xshift = 3.4cm]
    \node[draw=black, fill = red, fill opacity = 0.2, text opacity=1, circle, minimum size=6mm, inner sep=0pt] (y31) at (0,-1.4){$v^2_3$};
\end{scope}

\begin{scope}[xshift = 5.1cm]
    \node[draw=black, fill = red, fill opacity = 0.2, text opacity=1, circle, minimum size=6mm, inner sep=0pt] (y41) at (0,-1.4){$v^2_4$};
\end{scope}
\begin{scope}[xshift = 6.8cm]
    \node[draw=black, fill = red, fill opacity = 0.2, text opacity=1, circle, minimum size=6mm, inner sep=0pt] (y51) at (0,-1.4){$v^2_5$};
\end{scope}
\end{scope}

\draw[very thick] (x1.south) -- (y11.north);
\draw[very thick] (x1.south) -- (y21.north);
\draw[very thick] (x1.south) -- (y31.north);

\draw[very thick] (x2.south) -- (y11.north);
\draw[very thick] (x2.south) -- (y21.north);
\draw[very thick] (x2.south) -- (y31.north);
\draw[very thick] (x2.south) -- (y41.north);

\draw[very thick] (x3.south) -- (y11.north);
\draw[very thick] (x3.south) -- (y21.north);
\draw[very thick] (x3.south) -- (y31.north);

\draw[very thick] (x4.south) -- (y21.north);
\draw[very thick] (x4.south) -- (y41.north);
\draw[very thick] (x4.south) -- (y51.north);

\draw[very thick] (x5.south) -- (y41.north);
\draw[very thick] (x5.south) -- (y51.north);
    
\end{scope}

\end{tikzpicture}
    \caption{A construction of instance $(G, I_s,I_t, k')$.
    \emph{(Left.)} An instance of \prb{Dominating Set}, where $G$ consists of $5$ vertices and $5$ edges, and $k=2$. There is a dominating set $\{v_1,v_4\}$ of size $2$.
    \emph{(Right.)} \revrev{A simple description of resulting instance. A dotted rectangle indicates a clique. We highlight $I_s$ and $I_t$ by circles filled with red and colored by blue, respectively.}
    }
    \label{fig:split}
\end{figure}
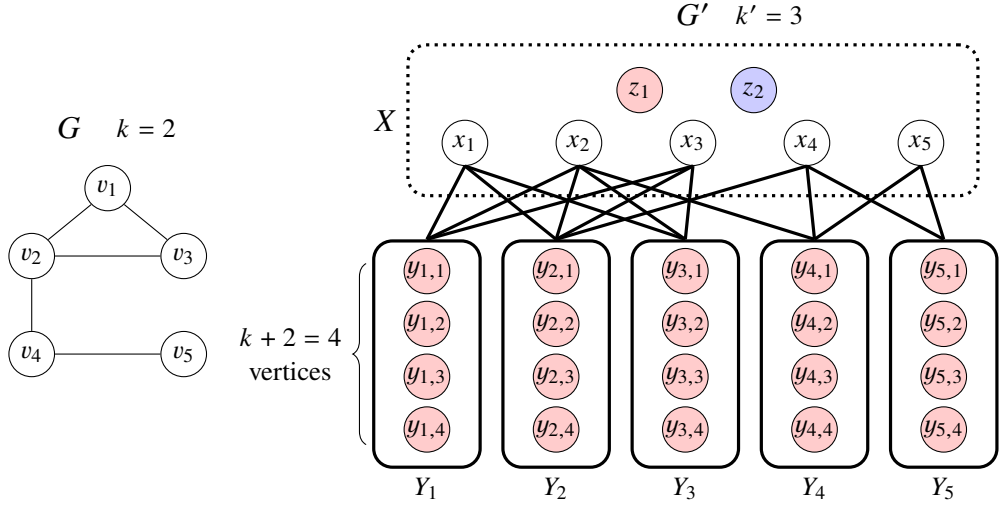

We first show the only-if direction.
Suppose that $G$ has a dominating set $D$ of size at most $k$.
Without loss of generality, assume that $D = \{v_1, v_2, \dots, v_\ell \}$, where $\ell \in [k]$. 
We define $I_0 = V_2 \cup \{v_0^1\}$ and $I_i = (I_{i-1}\setminus N(v_i^1)) \cup \{ v_i^1 \}$ for each $i \in [\ell]$.
Consider a sequence $\sigma = \langle I_s = I_0, I_1, \dots, I_\ell, I_{\ell+1} = I_t \rangle$ of length $\ell+1 \le k+1 \le k'$.
The following two lemmas establish that $\sigma$ is a solution to $(G', I_s, I_t, k')$.

\begin{lemma}\label{lem:splitright_ind}
    For $i \in \{0,\dots, \ell+1\}$, $I_i$ is an independent set of $G'$.
\end{lemma}
\iftrue
\begin{proof}
\revrev{
    The proof proceeds by induction on $i$. 
    By the construction of $G'$, $I_0 = V_2 \cup \{v_0^1\}$ is an independent set of $G'$.
    Suppose that $I_{i-1}$ is an independent set of $G'$ for $i \in [\ell]$.
    Clearly, $I_{i-1}\setminus N(v_i^1)$ is an independent set of $G'$.
    The vertex $v_i^1$ has no neighbor in $I_{i-1}\setminus N(v_i^1)$.
    Thus, $I_i = (I_{i-1}\setminus N(v_i^1)) \cup \{ v_i^1 \}$ is also an independent set of $G'$.
    Moreover, $I_{\ell+1} = I_t = \{v_0^1\}$ is obviously an independent set.
    }
\end{proof}
\fi

\begin{lemma}[$\ast$]\label{lem:splitright_con}
    For $i \in [\ell+1]$, the induced subgraph $G[I_{i-1} \bigtriangleup I_i]$ is connected.
\end{lemma}
\iftrue
\begin{proof}
\revrev{First, consider the case $i \in [\ell]$.
    Recall that $I_{i-1} \bigtriangleup I_i = (I_{i-1} \setminus I_i) \cup (I_{i} \setminus I_{i-1})$.
    Since $I_{i-1} \setminus I_i \subseteq N(v_i^1)$ and $I_{i} \setminus I_{i-1} = \{ v_i^1 \}$, we have $I_{i-1} \bigtriangleup I_i \subseteq N[v_i^1]$.
    Therefore, $G[I_{i-1} \bigtriangleup I_i]$ is connected for $i \in [\ell]$.
    }

    \revrev{Consider next the case $i = \ell+1$.
    By definition, it follows that $I_\ell = (V^2 \setminus \bigcup_{j \in [\ell]}N(v_j^1) ) \cup \{ v_\ell^1 \}$.
    Here, recall that $D$ is a dominating set of $G$, that is, $\bigcup_{i\in [\ell]} N[v_i] = V(G)$.
    Moreover, if $v_j \in N[v_i]$ for integers $i,j \in [n]$, then $v_i^1$ and $v_j^2$ are joined by an edge.
    This implies that $\bigcup_{j \in [\ell]}N(v_j^1) \subseteq V_2$ and hence $I_\ell = \{ v_\ell^1 \}$.
    We thus have $I_{\ell} \bigtriangleup I_{\ell+1} = \{v_0^1, v_\ell^1 \}$, and we conclude that $G[I_{\ell} \bigtriangleup I_{\ell+1}]$ is connected.}
\end{proof}
\fi

Next, we show the if-direction.
Suppose that there exists a sequence $\sigma = \langle I_s = I_0, I_1, \dots, I_\ell, I_{\ell+1} = I_t \rangle$ that is a solution to $(G', I_s, I_t, k')$, where $\ell \in [k]$ and hence $\ell+1 \le k+1 = k'$.
We here provide the following lemma.

\revrev{
\begin{lemma} \label{lem:split_ifdirection}
    For any $i \in [n]$, there exists an integer $j \in [\ell]$ such that at least one of the following two conditions holds: \emph{(1)} $N(v_i^2) \cap I_j \neq \emptyset$, or \emph{(2)} $I_{j-1} \bigtriangleup I_j = \{ v_i^2 \}$. 
\end{lemma}
\begin{proof}
    The proof proceeds by contradiction.
    Assume that $N(v_i^2) \cap I_j = \emptyset$ and $I_{j-1} \bigtriangleup I_j \neq \{ v_i^2 \}$ for every $j \in [\ell]$.
    The condition $I_{\ell} = I_{t} = \{v_0^1\}$ ensures the existence of an integer $r \in [\ell]$ such that $v_{i}^2 \in I_{r-1}$ and $v_{i}^2 \notin I_{r}$, which leads to $v_{i}^2 \subseteq I_{r-1} \bigtriangleup I_r$.
    Since $I_{r-1} \bigtriangleup I_r \neq \{ v_i^2 \}$ and $G[I_{r-1} \bigtriangleup I_{r}]$ is connected, it follows that there exists a vertex $v \in I_{r-1} \bigtriangleup I_{r}$ with $v \in N(v_i^2)$.
    However, this implies that $v \in I_{r-1}$ or $v \in I_{r}$, which contradicts $N(v_i^2) \cap I_j = \emptyset$ for every $j \in [\ell]$.
\end{proof}
}
\revrev{We define $D_1 = \{ v_i^1 \in V_1 \mid v_i^1 \in I_j \setminus I_{j-1}, j \in [\ell]\}$ and $D_2 = \{ v_i^2 \in V_2  \mid I_{j-1} \bigtriangleup I_j = \{ v_i^2 \}, j \in [\ell] \}$.
Note that $|I_j \cap D_1| \le 1$ for all $j \in [\ell]$ since $I_j$ is an independent set and $V_1$ is a clique.
We also have $|I_j \cap D_2| \le 1$ for all $j \in [\ell]$.
We here claim that $|D_1 \cup D_2| \le \ell$.
For the sake of contradiction, assume that $|D_1 \cup D_2| > \ell$.
Since $|I_j \cap D_1| \le 1$ and $|I_j \cap D_2| \le 1$ for all $j \in [\ell]$, there exist an integer $j \in [\ell]$, a vertex $v^1 \in I_j \cap D_1$, and a vertex $v^2 \in I_j \cap D_2$.
Consequently, we have $v^1 \in I_j \setminus I_{j-1}$; however, this yields $v^1 \in I_{j-1} \bigtriangleup I_j$, which contradicts the fact that $I_{j-1} \bigtriangleup I_j = \{ v^2 \}$.
}

\begin{lemma}\label{lem:split_D}
    $D$ is a dominating set of $G$.
\end{lemma}
\iftrue
\begin{proof}
    \revrev{It suffices to show that, for each $i \in [n]$, there exists a vertex in $N[v_i] \cap D$. 
    In the case where $v_i^2 \in D_2$, it follows immediately from the definition of $D$ that $v_i \in D$. 
    Thus, $v_i \in N[v_i] \cap D$ holds trivially.
    Suppose that $v_i^2 \notin D_2$.
    By \Cref{lem:split_ifdirection}, there exists an integer $j \in [\ell]$ such that $N(v_i^2) \cap I_j \neq \emptyset$. 
    Consider a vertex $v_r^1 \in N(v_i^2) \cap I_j$ for some $r \in [n]$. 
    Since $G'$ is a split graph, we have $v_r^1 \in V_1$.
    The fact that $v_r^1 \notin I_s$ ensures that $v_r^1 \in D_1$, which yields $v_r \in D$.
    Furthermore, by the construction of $G'$, the adjacency of $v_r^1$ and $v_i^2$ in $G'$ implies that $v_r$ and $v_i$ are adjacent in $G$. 
    Consequently, we obtain $v_r \in N[v_i] \cap D$.}
\end{proof}
\fi

It is known that \textsc{Dominating Set} is \W[2]-hard when parameterized by $k$~\cite{LokshtanovMPRS13}.
Moreover, \prb{Dominating Set} on graphs with $n$ vertices cannot be approximated within a factor of $(1-\delta)\ln n$ in polynomial time for any fixed constant $0<\delta<1$, unless $\P=\NP$~\cite{inapprox/DinurS14}, which follow from those for \prb{Set Cover}, a problem equivalent to \prb{Dominating Set} under approximation-preserving reductions~\cite{kann1992approximability}.
The reduction described above is an FPT-reduction with respect to $k$, immediately yielding the parameterized hardness result below.
\begin{corollary}
    \prb{SymISR} on split graphs is \W[2]-hard when parameterized by $k$.
\end{corollary}
\revrev{Moreover, although it does not directly preserve the approximation ratio, we show that it also yields the inapproximability result.}

\begin{lemma}\label{prbinapprox}
    \prb{SymISR} on split graphs with $n$ vertices cannot be approximated within a factor of $(1-\varepsilon)\ln n$ in polynomial
time for any fixed constant $\varepsilon >0$, unless $\P=\NP$.
\end{lemma}
\begin{proof}
    \revrev{
    The proof proceeds by contradiction that there is a polynomial-time $\rho(N)$-approximation algorithm, where $N$ is the number of vertices of a given graph for \prb{SymISR} and $\rho(N)\coloneq (1-\varepsilon)\log N$ for a fixed positive $0<\varepsilon<1$.
    Suppose that the algorithm outputs a reconfiguration sequence of length $L$ from $I_s$ to $I_t$ for the graph $G'$ obtained from the graph $G$ of \prb{Dominating Set}.}
    
    \revrev{We start with noting that $|V(G')|=2n+1$, where $n\coloneq |V(G)|$.
    Let $\gamma$ be the optimal size of the dominating set.
    By \Cref{lem:equivalence_dom_SymISR}, the optimal length of the reconfiguration sequence is $\gamma+1$.
    Hence, we have $L\le \rho(|V(G')|)(\gamma+1)$.
    Moreover, by \Cref{lem:equivalence_dom_SymISR}, we can construct the dominating set of size $|D|\le\rho(|V(G')|)(\gamma+1)-1$ for $G$.
    Thus, the ratio $|D|/\gamma$ is upper bounded by 
    \begin{align*}
    \frac{|D|}{\gamma} \le \frac{(\rho(|V(G')|)(\gamma+1)-1)}{\gamma}\le \left(1+\frac{1}{\gamma}\right)\rho(2n+1)\le \left(1+\frac{1}{\gamma}\right)(1-\varepsilon)\log (2n+1).
    \end{align*}
    Here, we can assume that $\gamma$ is lower bounded by a sufficiently large constant $c$, since the existence of a dominating set of size at most $c$ can be verified in polynomial time.
    Moreover, there is a small constant $c'$ such that $\log (2n+1)\le \log n + c'$.
    Hence, we obtain
    \begin{align*}
        \frac{|D|}{\gamma}\le \left(1+\frac{1}{c}\right)\left(1+\frac{c'}{\log n}\right)(1-\varepsilon)\log n
    \end{align*}
    using $c$ and $c'$.}

    \revrev{We now upper bound the value.
    First, since $\left(1+\frac{1}{c}\right)\left(1+\frac{c'}{\log n}\right)$ is strictly decreasing in both $c$ and $n$ and converges to $1$ as $c,n\to+\infty$, for every positive $\eta>0$, we can choose
    a sufficiently large constant $c$ and a constant $n_0$ such that $\left(1+\frac{1}{c}\right)\left(1+\frac{c'}{\log n}\right)\le 1+\eta$ for every $n>n_0$.}
    
    \revrev{Choose $\eta$ such that $0<\eta<\varepsilon/(1-\varepsilon)$, and we set
    $\delta \coloneq 1-(1+\eta)(1-\varepsilon)$.
    Then, we have $0<\delta<1$ and $(1+\eta)(1-\varepsilon)=1-\delta$.
    Hence, for $\gamma>c$ and $n>n_0$, we have
    \begin{align*}
        \frac{|D|}{\gamma}\le
        \left(1+\frac{1}{c}\right)
        \left(1+\frac{c'}{\log n}\right)
        (1-\varepsilon)\log n \le
        (1+\eta)(1-\varepsilon)\log n
        =
        (1-\delta)\log n.
    \end{align*}}

\revrev{It remains to consider the cases where $\gamma\le c$ or $n\le n_0$.
Since $c$ is a fixed constant, we can enumerate all vertex subsets of
size at most $c$ in polynomial time.
If such a dominating set exists, we can find an optimal one by exhaustive search; otherwise, we have $\gamma>c$ and apply the reduction above.
Moreover, if $n\le n_0$, an optimal dominating set can be found by
exhaustive search, since $n_0$ is a fixed constant.
Consequently, the $\rho(N)$-approximation algorithm for
\prb{SymISR}, together with these exhaustive searches, yields a
polynomial-time $(1-\delta)\log n$-approximation algorithm for
\prb{Dominating Set} on every instance.
This contradicts the $(1-\delta)\log n$-inapproximability of
\prb{Dominating Set} for every fixed positive $0<\delta<1$, unless $\P=\NP$.}
\end{proof}

We also remark that, due to the minimality of the dominating sets, the subsets $(I_i \cap Y)$ in the solution $\langle I_s=I_0, I_1, \dots, I_\ell=I_t\rangle$ for $(G', I_s, I_t, k')$ form a strictly decreasing sequence with respect to set inclusion.
Formally, the following proposition is obtained.
\begin{proposition}\label{prop:monotone}
    Let $(G', I_s, I_t, k')$ be the constructed instance of \prb{SymISR} as above, where $G'$ is a split graph.
    Then, there exists a sequence $\sigma = \langle I_s = I_0, I_1, \dots, I_{\ell-1}, I_{\ell} = I_t \rangle$ such that $\sigma$ is a solution to $(G', I_s, I_t, k')$ and $I_{i-1} \cap Y \supset I_{i} \cap Y$ for every $i \in [\ell]$.
\end{proposition}

\begin{proof}[Proof of \Cref{monotoneinapprox}]
    Combining \Cref{prbinapprox} and the discussion in \Cref{polytopeinterpretation}, the inapproximability of finding the shortest path from $u$ to $v$ on $P$ is readily shown. 
    \revrev{The remaining task is finding a suitable linear function $c$ to ensure that the corresponding path on $P$ is monotone with respect to at least one objective.
    To this end, we consider the objective $c: \{0,1\}^{V(G')}\to \mathbb{R}$ such that $c(\bm{x}) = x_{v^1_0}-2\sum_{v\in V^2} x_v$.
    We remark that $I_t=\{v^1_0\}$ is the only independent set with positive objective value among all independent sets of $G'$.}
    
    We now show that every reconfiguration sequence of the shortest length strictly increases the objective value.
    \revrev{Here, we note that the second term $-2\sum_{v\in V^2} x_v$ of $c$ is strictly increasing on the shortest path, by \Cref{prop:monotone}.
    Moreover, for every integer $i\in [\ell]$, we have $v^1_0\notin I_{i}$; if $v^1_0\in I_i$, then we have $I_i\cap X=\{v^1_0\}$ and $I_{i-1}\bigtriangleup I_i$ must be $\{x_{i-1}, v^1_0\}$.
    This contradicts \Cref{prop:monotone}, since $I_{i-1}\cap V^2=I_i\cap V^2$.
    Thus, the objective value corresponds to $I_i$ can be computed as: $1-2|V^2|$ when $i=0$, $-2|I_i\cap V^2|$ when $i\in [1,\ell-1]$, $0$ when $i=\ell$, and $1$ when $i=\ell+1$.}
    It is clear that the value is strictly increasing, completing the proof of \Cref{monotoneinapprox}.
\end{proof}






\subsection{Block Graphs}\label{sec:block}

In this section, we show that \prb{SymISR} is solvable in polynomial time on block graphs.
We begin by introducing some notations to define block graphs.
A \emph{biconnected component} of $G$ is a maximal biconnected subgraph, that is, a subgraph $H$ such that the subgraph induced by $V(H)\setminus \{v\}$ is connected for every $v\in V(H)$.
A graph $G$ is \emph{block} if every biconnected component of $G$ forms a clique~\cite{graphclasses}.

\begin{theorem}\label{thm:block}
    \prb{SymISR} can be solved in $O(n^5)$ time for block graphs.
\end{theorem}

Here, for the simplicity of our algorithm, we inductively define \emph{rooted block graphs} by the following three operations.
A rooted block graph consists of a graph $G$ together with a \emph{root} $r(G)\in V(G)$, and a \emph{root clique}
$C(G)\subseteq V(G)$ containing $r(G)$.

\begin{enumerate}
    \item A graph $G=(\{v\}, \emptyset)$ with only one vertex is a rooted block
    graph with $r(G)=v$ and $C(G)=\{v\}$.

    \item Let $G_1$ and $G_2$ be two rooted block graphs with $V(G_1)\cap V(G_2)=\{r(G_1)\}=\{r(G_2)\}$.
    Then $G=(V(G_1)\cup V(G_2),E(G_1) \cup E(G_2))$ is a rooted block graph  with $r(G)= r(G_1)$ and $C(G)=\{r(G)\}$.

    \item  Let $G_1$ and $G_2$ be two rooted block graphs with $V(G_1)\cap V(G_2)=\emptyset$. 
    let $G$ be obtained from $G_1\oplus G_2$ by adding the edges $\{\, r(G_2)c \mid c\in C(G_1)\,\}$.
    Then $G$ is a rooted block graph with $V(G)=V(G_1)\cup V(G_2)$, $E(G)= E(G_1) \cup E(G_2)\cup \{\, r(G_2)c \mid c\in C(G_1)\,\}$, $r(G)=r(G_1)$ and, 
    $C(G)= C(G_1)\cup\{r(G_2)\}$.
\end{enumerate}

This definition naturally raises a decomposition tree, which forms a rooted binary tree, where each leaf corresponds to a single-vertex graph constructed in Operation 1, and each internal node corresponds to either Operation 2 or 3. 
An example of a decomposition tree for a block graph is shown in \Cref{fig:block}. 
We write $T(G)$ as a decomposition tree of $G$.

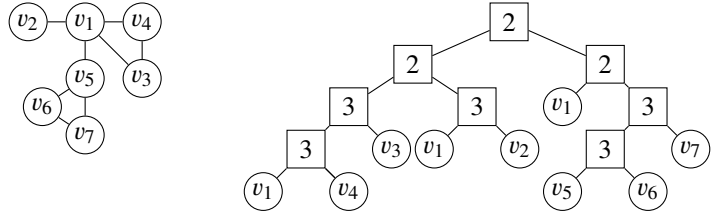
\begin{figure}
    \centering
    \begin{tikzpicture}[scale = 0.75]
\begin{scope}[xshift = -1cm]
\node[draw=black, circle, minimum size=5mm, inner sep=0pt] (v1) at (0,0){$v_1$};
\node[draw=black, circle, minimum size=5mm, inner sep=0pt] (v2) at (-1,0){$v_2$};
\node[draw=black, circle, minimum size=5mm, inner sep=0pt] (v3) at (1,-1){$v_3$};
\node[draw=black, circle, minimum size=5mm, inner sep=0pt] (v4) at (1,0){$v_4$};
\node[draw=black, circle, minimum size=5mm, inner sep=0pt] (v5) at (0,-1){$v_5$};
\node[draw=black, circle, minimum size=5mm, inner sep=0pt] (v6) at (-0.75,-1.5){$v_6$};
\node[draw=black, circle, minimum size=5mm, inner sep=0pt] (v7) at (0,-2){$v_7$};
\draw (v1) -- (v2);
\draw (v1) -- (v3);
\draw (v1) -- (v4);
\draw (v1) -- (v5);
\draw (v3) -- (v4);
\draw (v5) -- (v6);
\draw (v5) -- (v7);
\draw (v7) -- (v6);
    
\end{scope}
\begin{scope}[xshift = 7cm, scale=0.75]
\node[draw=black, circle, minimum size=5mm, inner sep=0pt] (v11) at (-6.5,-4){$v_1$};
\node[draw=black, circle, minimum size=5mm, inner sep=0pt] (v12) at (-2.5,-3){$v_1$};
\node[draw=black, circle, minimum size=5mm, inner sep=0pt] (v13) at (0.5,-2){$v_1$};
\node[draw=black, circle, minimum size=5mm, inner sep=0pt] (v2) at (-4.5,-4){$v_4$};
\node[draw=black, circle, minimum size=5mm, inner sep=0pt] (v3) at (-3.5,-3){$v_3$};
\node[draw=black, circle, minimum size=5mm, inner sep=0pt] (v4) at (-0.5,-3){$v_2$};
\node[draw=black, circle, minimum size=5mm, inner sep=0pt] (v5) at (0.5,-4){$v_5$};
\node[draw=black, circle, minimum size=5mm, inner sep=0pt] (v6) at (2.5,-4){$v_6$};
\node[draw=black, circle, minimum size=5mm, inner sep=0pt] (v7) at (3.5,-3){$v_7$};
\node[draw=black, rectangle, minimum size=5mm, inner sep=0pt] (b12) at (-5.5,-3){$3$};
\node[draw=black, rectangle, minimum size=5mm, inner sep=0pt] (b123) at (-4.5,-2){$3$};
\node[draw=black, rectangle, minimum size=5mm, inner sep=0pt] (b14) at (-1.5,-2){$3$};
\node[draw=black, rectangle, minimum size=5mm, inner sep=0pt] (b1234) at (-3,-1){$2$};
\node[draw=black, rectangle, minimum size=5mm, inner sep=0pt] (b1234567) at (-0.75,0){$2$};
\node[draw=black, rectangle, minimum size=5mm, inner sep=0pt] (b56) at (1.5,-3){$3$};
\node[draw=black, rectangle, minimum size=5mm, inner sep=0pt] (b567) at (2.5,-2){$3$};
\node[draw=black, rectangle, minimum size=5mm, inner sep=0pt] (b1567) at (1.5,-1){$2$};

\draw (v11) -- (b12);
\draw (v12) -- (b14);
\draw (v13) -- (b1567);
\draw (v2) -- (b12);
\draw (b123) -- (v3);
\draw (b14) -- (v4);
\draw (b56) -- (v5);
\draw (b56) -- (v6);
\draw (b567) -- (v7);

\draw (v2) -- (b12);
\draw (b123) -- (b12);
\draw (b123) -- (b1234);
\draw (b14) -- (b1234);
\draw (b1234567) -- (b1234);
\draw (b1234567) -- (b1567);
\draw (b567) -- (b1567);
\draw (b567) -- (b56);

\end{scope}

\end{tikzpicture}
    \caption{Construction of a decomposition tree for block graphs.}
    \label{fig:block}
\end{figure}

\begin{lemma}\label{lem:blockdecomposition}
    Let $G=(V,E)$ be a connected block graph and $r\in V$. Then we can compute a decomposition tree $T(G)$ for $G$ with $r(G)=r$ and $O(n)$ nodes in polynomial time. For this decomposition holds that $r$ is in exactly one maximal clique or $C(G)=\{r\}$.
\end{lemma}
\iftrue
\begin{proof}
\revrev{We first prove the existence inductively on the number of vertices. So, let $G=(V(G),E(G))$ be a graph and $r\in V(G)$ a vertex. If $V(G)=\{r\}$ then the statement is true. Therefore, assume $\vert V(G)\vert >1$. First, we consider the case where $r$ is in multiple maximal cliques $K_1,\ldots, K_t$. In this case, $G[V(G)\setminus \{r\}]$ is no longer connected.  Otherwise, $K_1\cup \ldots \cup K_t$ would be a block (so clique) and $K_t$ would not be minimal in the first place. Let $C_t$ be the connected component of $G[V(G)\setminus \{r\}]$ including $K_t$. Then $G'=G[C_t\cup \{r\}]$ and $G''=G[V(G) \setminus C_t]$ are connected block graphs with $\vert V(G')\vert, \vert V(G')\vert < \vert V(G)\vert$. In this case, we can use an inductive argument together with Operation 2 (implying $C(G)=\{r\}$). Hence, we can assume $r$ is in only one maximal clique $C(G)$. Let $v \in C(G) \setminus \{r\}$. If $C(G)$ is the only maximal clique including $v$ then we can us Operation 1 to construct $G_2=(\{v\},\emptyset)$ and Operation 3 with $G_1=G[V(G)\setminus \{v\}]$. If $v$ is in multiple cliques $G'=G[V(G)\setminus \{v\}]$  is not connected. Let $G_1$ be the connected component including $r$. Define $G_2=[V(G)\setminus V(G_1)]$. By using Operation 3 on $G_1,G_2$ with $r(G_1)=r$ (we can assume this by an inductive argument), we obtain $G$ with $r=r(G_1)=r(G)\in C(G)=C(G_1)\cup \{v\}$.}

   \revrev{Note that the number of nodes of $T(G)$ can be bounded by $O(n)$, by the following observation.
    Let $\mathcal{B}$ be the set of biconnected components of $G$, and define $b\coloneq |\mathcal{B}|$.
    First, each leaf node corresponds to a vertex $v$ in some biconnected component $B\in \mathcal{B}$, thus the number of leaf nodes is exactly $\sum_{B\in \mathcal{B}}|B|=n+b-1$.
    Next, the number of times Operation 2 is used can be bounded by the number of identifications of two vertices, hence the number  $b$ of biconnected components.
    Finally, the number of times Operation 3 is used can be bounded by the number of vertices.
    It is clear that $b$ can be bounded by $n-1$, and hence we can conclude that there is $O(n)$ nodes in $T(G)$. Observe that the maximal cliques of $G$ can be computed in polynomial time. Since all steps in the induction are simple case distinctions (depending on how many vertices are incident with $r$), the running time is bounded polynomial in $T(G)$.}
\end{proof}
\fi


We begin with the following observations.
Let $G$ be a rooted block graph.
Recall that, when two independent sets $I$ and $J$ are adjacent, the operation that transforms $I$ into $J$ can be viewed as a flip on a connected vertex set $X = I \bigtriangleup J$.
For a reconfiguration sequence $\langle I_s=I_0,I_1,\dots,I_\ell=I_t\rangle$, consider three consecutive independent sets $I_i,I_{i+1},I_{i+2}$ of $G$ such that
$I_i\cap C=\{u\}$, $I_{i+1}\cap C=\emptyset$, and $I_{i+2}\cap C=\{v\}$, where $u$ and $v$ are two distinct vertices of a maximal clique $C$ of $G$.
Let $X_i = I_i \bigtriangleup I_{i+1}$ and $X_{i+1} = I_{i+1} \bigtriangleup I_{i+2}$.
Observe that $X_i \cap C = \{ u \}$ and $X_{i+1} \cap C = \{ v \}$.
Since $G$ \rev{is} a rooted block graph, $X_i$ and $X_{i+1}$ are disjoint.
Moreover, the fact that $uv \in E(G)$ ensures that $G[X_i \cup X_{i+1}]$ is connected.
This implies that $I_{i+2}$ is obtained from $I_i$ by a single flip on $X_i \cup X_{i+1}$, without passing through $I_{i+1}$.
We call this operation a \emph{synchronized flip}.
Note that a synchronized flip at $C$ cannot be performed simultaneously with another flip involving a vertex of the same clique, since every set in the reconfiguration sequence is independent and hence contains
at most one vertex of $C$.

Moreover, consider a subsequence $I_i,I_{i+1},\ldots,I_j,I_{j+1}$ such that $I_i\cap C=\{u\}$, $I_{i'}\cap C=\emptyset$ for every $i'\in [i+1,j]$, and $I_{j+1}\cap C=\{v\}$.
Then, for each $i'\in [i,j]$, $I_{i'+1}$ can be obtained from $I_{i'}$ by flipping $X_{i'} = I_{i'} \bigtriangleup I_{i'+1}$.
Let $G'$ be the graph obtained from $G$ by removing $C$.
Observe that there is a connected component of $G'$, say $H$, such that $X_i \subseteq V(H) \cup \{ u \}$.
If there is an integer $i'\in [i+1,j-1]$ such that $X_{i'} \cap V(H) = \emptyset$, then the flip on $X_{i'}$ can be performed before the flip on $X_i$.
This claim is justified because the flips on $X_{i}$ and $X_{i'}$ can be done independently, \rev{since $X_{i'}$ is contained in a different connected component from $H$}.
Therefore, the order of the flips can be rearranged while ensuring that the sequence of independent sets generated by these flips remains a reconfiguration sequence:
performing flips on $X_{i'}$ for all $i'\in [i+1,j-1]$ satisfying $X_{i'} \cap V(H) = \emptyset$, followed by the synchronized flip on $X_i \cup X_j$, and finally flips on $X_{i'}$ for all $i'\in [i+1,j-1]$ satisfying $X_{i'} \subseteq V(H)$.
Consequently, \rev{any flip involving exactly one vertex of $C$ can be synchronized with the subsequent flip involving exactly one vertex of $C$.}

A similar discussion can be applied to the root $r(G)$.
Let $H_1$ and $H_2$ be two connected components in the graph obtained from $G$ by removing $r(G)$.
For vertex sets $X_1 \subseteq V(H_1) \cup \{ r(G)\}$ and $X_2 \subseteq V(H_2) \cup \{ r(G)\}$, if $X_1 \cap X_2 = \{ r(G) \}$, then the flips on $X_1$ and $X_2$ can be \emph{merged}.

We now briefly describe the strategy of our algorithm.
Suppose that a rooted block graph $G$ is obtained from rooted block graphs $G_1$ and $G_2$ by Operation~2 or Operation~3.
Several flips that are counted separately in $G_1$ and $G_2$ can be synchronized or merged into a single step in $G$.
\rev{
Our algorithm records information about flips that may later be synchronized or merged on each rooted block graph constructed by a subtree in $T(G)$.
It then processes the rooted block graph in a bottom-up manner and accounts for possible synchronization or merging when Operation~2 or Operation~3 is applied.
}
If $G$ is obtained by Operation~2, then $G$ is formed by identifying $r(G_1)$ and $r(G_2)$.
Hence, only flips involving $r(G) = r(G_1) = r(G_2)$ can be merged.
Accordingly, we record the number of flips involving the root.
If $G$ is obtained by Operation~3, then a flip involving exactly one vertex of $C(G_1)$ can be synchronized with a flip involving $r(G_2)$.
We record the number of flips involving $C(G)$ that remain available for synchronization.
As discussed above, for a reconfiguration sequence $\rho$, this number equals to the number of independent sets $I_i$ in $\rho$ satisfying $|I_i\cap C(G')|=1 \land I_{i+1}\cap C(G')=\emptyset$ or $I_i\cap C(G')=\emptyset \land |I_{i+1}\cap C(G')|=1$.

Formally, for each rooted block graph $G'$ constructed by a subtree in $T(G)$ and each pair of integers $p,q \in [0,n]$, the value $f(G',p,q)$ represents the shortest length over all reconfiguration sequences on $G'$ such that exactly $p$ flips involve the root $r(G')$, and exactly $q$ flips are left to be synchronized. If no such sequence exists, we set $f(G', p, q) = +\infty$.
The value $f(G', p, q)$ can be computed using dynamic programming in a bottom-up manner on $T(G)$.
\rev{We omit the discussion for analysis of the running time.}

\iftrue

\subsubsection{Operation 1: base case.}
Consider the case when a graph $G$ consists of a singleton ${v}$.
Suppose that there is a shortest reconfiguration sequence $\rho$ of $G$ such that flips on $r(G) = v$ occur exactly $p$ times.
By the minimality, the length of $\rho$ is exactly $p$.
Let $\odot$ denote the XNOR operator.
When $p$ is even, it holds that $v\in I_s \land v\in I_t =\TRUE$ or $v\notin I_s \land v\notin I_t =\TRUE$, that is, $v\in I_s \odot v\in I_t$ must be $\FALSE$.
On the other hand, when $p$ is odd, $v\in I_s \odot v\in I_t$ must be $\TRUE$.
Consequently, we have $(p \text{ is even})\odot (v\in I_s\odot v\in I_t)=\TRUE$.
Moreover, an index $q$ must be equal to $p$, since $C(G)=\{v\}$ and hence every flip involving $v$ is counted toward future synchronized flips.
Conversely, if $(p \text{ is even})\odot (v\in I_s\odot v\in I_t)=\FALSE$ and $p = q$, such a reconfiguration sequence $\rho$ can be constructed.
We thus obtain the following.
\begin{align*}
    f(G,p,q)=\begin{cases}
        p & \text{ if } (p \text{ is even})\odot (v\in I_s\odot v\in I_t)=\TRUE ~\land~ p = q, \\
        +\infty & \text{otherwise.}
    \end{cases}
\end{align*}

\subsubsection{Operation 2.}\label{sec:operation2}
Next, we consider the case where $G$ is obtained from two rooted block graphs
$G_1$ and $G_2$ by Operation~2.
To compute $f(G,p,q)$, consider a shortest reconfiguration sequence $\rho$ such that flips on $r(G)$ occur $p$ times.
Since $r(G) = r(G_1) = r(G_2)$, the roots $r(G_1)$ and $r(G_2)$ must be flipped exactly $p$ times in
both $G_1$ and $G_2$, respectively.
On the other hand, the number of indices $i$ satisfying
$
|I_i\cap C(G_1)|=1 \land I_{i+1}\cap C(G_1)=\emptyset,
$ or $
I_i\cap C(G_1)=\emptyset \land |I_{i+1}\cap C(G_1)|=1,
$
and the analogous condition for $G_2$, may be arbitrary.
Thus, when combining the two sequences, we may use any values
$f(G_1,p,q_1)$ and $f(G_2,p,q_2)$ with $q_1,q_2\in [n]$.
Since $C(G)=\{r(G)\}$, every flip involving $r(G)$ is counted toward future
synchronized flips by the same argument as in the base case, and hence $f(G, p,q)$ has a finite value only if $p=q$.
Finally, although $f(G_1,p,q_1)$ and $f(G_2,p,q_2)$ count the flips involving
$r(G_1)$ and $r(G_2)$ separately, these correspond to the same $p$ flips once
$r(G_1)$ and $r(G_2)$ are identified in $G$.
Hence, $p$ overcounts arise, and we subtract $p$ accordingly.
Thus, we now have the following.
\begin{align*}
    f(G,p,q)=\begin{cases}
        +\infty & \text{ if } p\ne q\\
        \makecell[l]{\min_{q_1,q_2\in [0,n]} \{f(G_1,p,q_1)+f(G_1, p, q_2)-p\}\\
        =\min_{q_1\in [0,n]}f(G_1,p,q_1) + \min_{q_2\in [0,n]} f(G_2,p,q_2)-p}&\text{ otherwise.}
    \end{cases}
\end{align*}

\subsubsection{Operation 3.}\label{sec:operation3}
Consider the case where $G$ is obtained from two rooted block graphs
$G_1$ and $G_2$ by Operation~3.
We now explain how to compute $f(G,p,q)$ from the already computed values
$f(G_1,p_1,q_1)$ and $f(G_2,p_2,q_2)$, where
$p_1,p_2,q_1,q_2\in [0,n]$.
Let $\sigma_{G_1}$ and $\sigma_{G_2}$ be shortest reconfiguration sequences witnessing
$f(G_1,p_1,q_1)$ and $f(G_2,p_2,q_2)$, respectively.
First, every flip involving $r(G)$ corresponds exactly to a flip involving $r(G_1)$, and vice versa.
Hence, we must have $p=p_1$.

We consider joining $\sigma_{G_1}$ and $\sigma_{G_2}$ into a single reconfiguration sequence $\sigma_G=\langle I_0,I_1,\dots,I_\ell\rangle$ of $G$.
We define binary sequences $\flipC$ and $\flipr$ of lengths
$|\sigma_{G_1}|+1$ and $|\sigma_{G_2}|+1$, respectively.
For each $i\in\{0,\dots,|\sigma_{G_1}|\}$, $\flipC[i] =\TRUE$ if and only if $|\sigma_{G_1}[i]\cap C(G_1)|=1$.
Similarly, for each $j\in\{0,\dots,|\sigma_{G_2}|\}$, $\flipr[j] =\TRUE$ if and only if $r(G_2)\in \sigma_{G_2}[j]$.
By the definition of $\flipC[0]$ and $\flipr[0]$, we have $(\flipC[0], \flipr[0])\ne (\TRUE, \TRUE)$.
Here, for a sequence $A$ of length $|A|$ and integers $a,b$ with $0\le a\le b\le |A|-1$, let $A[a:b]$ denote the subsequence consisting of the $(a+1)$-th through $(b+1)$-th elements of $A$.

We use the following simple observation.
Assume that $q_1\le p_2$ and $\flipC[0]\ne \flipr[0]$.
By the definition of $q_1$, there are exactly $q_1$ indices $i$ such that either $|I_i\cap C(G_1)|=1$ and $|I_{i+1}\cap C(G_1)|=0$,
or
$
|I_i\cap C(G_1)|=0
$ and $
|I_{i+1}\cap C(G_1)|=1.
$
All these flips can be synchronized with flips involving $r(G_2)$.
Since the symmetry also holds, we have the following.

\begin{proposition}\label{prop:otoku}
If $\flipC[0]\ne \flipr[0]$, then the number of possible synchronized flips is $\min\{q_1,p_2\}$.
\end{proposition}
We now show how to compute the value $f(G,p,q)$, by separately showing the equation for three cases: $q_1>p_2$, $q_1<p_2$, and $q_1=p_2$.

\paragraph*{Case: $q_1>p_2$.}
Let $g_1(G,p,q)$ denote the minimum length of a reconfiguration sequence for $G$ in the case $q_1>p_2$.
First, if $p_2=0$ and $r(G_2)\in I_s$, then any choice of $q_1>0$ yields no reconfiguration sequence, since no vertex of $C(G_1)$ is contained in any independent set.

We distinguish two cases according to the initial values of $\flipC$ and $\flipr$:
either $\flipC[0]\ne \flipr[0]$, or $\flipC[0]=\flipr[0]=\FALSE$.
Recall that $(\flipC[0], \flipr[0])\ne (\TRUE, \TRUE)$ by the definition.

If $\flipC[0]\ne \flipr[0]$, then \Cref{prop:otoku} applied to the pair $(\flipC,\flipr)$ implies that the number of synchronized flips is $\min\{q_1,p_2\}=p_2$.
Now assume that $\flipC[0]=\flipr[0]=\FALSE$.
Here, let $i$ be the smallest index such that $\flipC[i]=\TRUE$.
Then the first value $\flipC[i]$ and $\flipr[0]$ of the two sequences
are different.
Moreover, the truth value in $\flipC[i:|\sigma_{G_1}|]$ changes exactly $q_1-1$ times, whereas the truth value in $\flipr[0:|\sigma_{G_2}|]$ changes exactly $p_2$ times.
Since $q_1>p_2$, we have $q_1-1\ge p_2$.
Hence, by \Cref{prop:otoku}, the number of possible synchronized flips is
$\min\{q_1-1,p_2\}=p_2$.
\rev{Now, to compute the value $g_1(G,p,q)$, we introduce an auxiliary positive integer $s\in [0,p_2]$, which indicates the number of synchronized flips.}
The number of indices $i$ satisfying either $|I_i\cap C(G)|=1 \land I_{i+1}\cap C(G)=\emptyset$ or $I_i\cap C(G)=\emptyset \land |I_{i+1}\cap C(G)|=1$ is exactly $(q_1-s)+(p_2-s)=q_1+p_2-2s$ in both cases, since the numbers of possible synchronized flips in $C(G_1)$ and $r(G_2)$ decrease by $s$ in both cases.
Therefore, we consider pairs $(q_1,p_2)$ satisfying $q=q_1+p_2-2s$.
Moreover, the length of the reconfiguration sequence for $G$ is reduced by $s$ from $f(G_1, p_1, q_1) + f(G_2, p_2, q_2)$.

Therefore, the value $g_1(G,p,q)$ can be computed as follows.
For each $p_2\in[0,n]$, we compute
$f^{*}(G_2,p_2)\coloneq \min_{q_2\in[0,n]} f(G_2,p_2,q_2)$,
since $q_2$ may take any value.
For this value, we have:
\begin{align*}
    g_1(G,p,q)=\begin{cases}
         \min_{\substack{1\le p_2< q_1\le n\\s\in [0,p_2]\\q_1+p_2-2s=q}} \{f(G_1,p,q_1)+f^{\ast}(G_2,p_2)-s\}& \text{ if } r(G_2)\in I_s\\
         \min_{\substack{0\le p_2< q_1\le n\\s\in [0,p_2]\\q_1+p_2-2s=q}} \{f(G_1,p,q_1)+f^{\ast}(G_2,p_2)-s\}& \text{ otherwise. }\\
    \end{cases}
\end{align*}

\paragraph*{Case: $q_1< p_2$.}
Let $g_2(G,p,q)$ denote the minimum length of a reconfiguration sequence for $G$ in the case $q_1<p_2$.
By symmetry with the case $q_1>p_2$, the number of possible synchronized flips is $q_1$, and the numbers of possible synchronized flips in $C(G_1)$ and $r(G_2)$ decrease by $s\in [0,q_1]$ in both cases.
Moreover, if $q_1=0$ and $|C(G_1)\cap I_s|=1$, then any choice of $p_2>0$ yields no reconfiguration sequence, since $r(G_2)$ is not contained in any independent set.
Thus, the value $g_2(G,p,q)$ can be computed as:
\begin{align*}
g_2(G,p,q)=\begin{cases}
         \min_{\substack{1\le q_1< p_2\le n\\s\in [0,q_1]\\p_2+q_1-2s=q}} \{f(G_1,p,q_1)+f^{\ast}(G_2,p_2)-s\}& \text{ if } |C(G_1)\cap I_s|=1\\
         \min_{\substack{0\le q_1< p_2\le n\\s\in [0,
         q_1]\\p_2+q_1-2s=q}} \{f(G_1,p,q_1)+f^{\ast}(G_2,p_2)-s\}& \text{ otherwise, }\\
    \end{cases}
\end{align*}
where $f^{*}(G_2,p_2)= \min_{q_2\in[0,n]} f(G_2,p_2,q_2)$.

\paragraph*{Case: $q_1= p_2$.}
Let $g_3(G,p,q)$ denote the minimum length of a reconfiguration sequence for $G$ in the case $q_1=p_2$.
We distinguish three cases according to the initial values of $\flipC$, $\flipr$, and $q_1$:
either $\flipC[0]\ne \flipr[0]$, $\flipC[0]=\flipr[0]=\FALSE\land q_1=0$, $\flipC[0]=\flipr[0]=\FALSE\land q_1>0$.
Recall that $(\flipC[0], \flipr[0])\ne (\TRUE, \TRUE)$ by the definition.
In each case, we count the number of possible synchronized flips.

First, if $\flipC[0]\ne \flipr[0]$, then we can apply \Cref{prop:otoku}, and the number of possible synchronized flips is $q_1$.
Note that, the condition $\flipC[0]\ne \flipr[0]$ can be translated to $(|C(G_1)\cap I_s|=1\land r(G_2)\notin I_s)\lor (|C(G_1)\cap I_s|=0\land r(G_2)\in I_s)$. 
Thus, either $|C(G_1)\cap I_s|=1$ or $r(G_2)\in I_s$ holds.
Since $C(G)= C(G_1)\cup\{r(G_2)\}$, the condition $\flipC[0]\ne \flipr[0]$ is equivalent to $|C(G)\cap I_s|=1$.

Next, if $\flipC[0]=\flipr[0]=\FALSE\land q_1=0$, then there is no possible synchronized flip.

Finally, consider the case $\flipC[0]=\flipr[0]=\FALSE\land q_1>0$.
We claim that the number of possible synchronized flips is $q_1-1$ in this case.
The condition $q_1=p_2>0$ ensures that each of $C(G_1)$ and $r(G_2)$ is intersected by at least one independent set in the sequence.
Furthermore, since there is no independent set $I$ of $G$ such that $|I\cap C(G_1)|=1$ and $r(G_2)\in I$, there exists the minimum integer $j\in[\ell]$ such that exactly one of $|I_j\cap C(G_1)|=1$ and $r(G_2)\in I_j$ holds, that is, exactly one of $\flipC[j]$ and $\flipr[j]$ is $\TRUE$.
If $\flipC[j]=\TRUE$, then we apply \Cref{prop:otoku} to $\flipC[j:|\sigma_{G_1}|]$ and $\flipr$.
If $\flipr[j]=\TRUE$, then we apply \Cref{prop:otoku} to $\flipr[j:|\sigma_{G_2}|]$ and $\flipC$.
Hence, in either case, the number of possible synchronized flips is $q_1-1$.

Let $s$ be the number of synchronized flips.
As discussed in Case $q_1 < p_2$, we have $q=q_1+p_2-2s=2(p_2-s)$, and the length of the reconfiguration sequence for $G$ is reduced by $s$ from $f(G_1, p_1, q_1) + f(G_2, p_2, q_2)$.
Therefore, $g_3(G,p,q)$ can be computed as follows.
\begin{align*}
    g_3(G,p,q)=\begin{cases}
      \min_{\substack{q_1\in [0,n]\\ s\in [0,q_1]\\ 2(q_1-s)=q}} \{f(G_1,p, q_1)+f^*(G_2,q_1)-s\}  &\text{ if } |C(G)\cap I_s|=1\\
      \min\left(\makecell[l]{f(G_1,p, 0)+f^*(G_2,0),\\\min_{\substack{q_1\in [1,n]\\ s\in [0,q_1-1]\\ 2(q_1-s)=q}} \{f(G_1,p, q_1)+f^*(G_2,q_1)-s\}}\right)&\text{ otherwise, }
    \end{cases}
\end{align*}
where $f^{*}(G_2,q_1)= \min_{q_2\in[0,n]} f(G_2,q_1,q_2)$.

Finally, $f(G,p,q)$ is computed by the minimum among $g_1(G,p,q),g_2(G,p,q),g_3(G,p,q)$.

\iftrue
\revrev{\subparagraph*{Running time.}
We now analyze the running time for this algorithm.
First, one can observe that there are $O(n^2)$ functions for each subgraph $G'$ constructed in a subtree in $T(G)$, and there are $O(n)$ nodes in the decomposition tree.
We proceed to estimate the time for computing $f(G',p,q)$ for each case in which Operations 1, 2, and 3 are applied.
\begin{itemize}
    \item Operation 1: Clearly, it can be done in constant time by checking each condition.
    \item Operation 2: Both $\min_{q_1\in [0,n]}f(G_1,p,q_1)$ and $\min_{q_2\in [0,n]} f(G_2,p,q_2)$ can be obtained in $O(n)$ time.
    Moreover, if $p\ne q$, then the algorithm simply returns a value $+\infty$; thus, the total time for computing $f(G,p,q)$ is $O(n)$.
    \item Operation 3: First, we can observe that $f^\ast(G_2, p_2)$ and $f^\ast(G_2, q_1)$ can be computed in $O(n)$ time for each $p_2\in [0,n]$ and $q_2 \in [0,n]$, respectively.
    Moreover, for fixed $q_1,p_2\in [0,n]$, the value $s$ is uniquely determined by $q_1+p_2-2s=q$.
    Thus, the time for computing $g_1(G,p,q)$, $g_2(G,p,q)$, and $g_3(G,p,q)$ are bounded by $O(n^2)$.
\end{itemize}
Therefore, the overall running time can be bounded by $O(n^2)\cdot O(n^2)\cdot O(n)=O(n^5)$.}
\fi
\fi

\subsubsection*{Trees.}

A graph is a tree if every biconnected component of $G$ forms a clique with two vertices. 
Thus, every tree is a block graph, and hence the polynomial-time solvability for trees follows immediately from \Cref{thm:block}.
We improve the running time to be faster than $O(n^5)$ through some observations.

Our core observation is that Operation 3 can only be applied when $C(G_1)=\{r(G_1)\}$ to obtain a tree, since joining $G_2$ into $G_1$ with $|C(G_1)|\ge 2$ results in a cycle consisting of at least three vertices.
Thus, all synchronized flips involving a vertex of $C(G)$ also involve $r(G)$, implying that they can be tracked using only by $r(G)$.

\iftrue
\revrev{For each rooted tree $G$ and each integer $p\in [n]$, let $h(G,p)$ be the
shortest length over all reconfiguration sequences on $G$ such that exactly
$p$ flips involving $r(G)$ can be used for future synchronized flips.}

\subparagraph*{Operation 1.}
\revrev{Although we define $f(G,p,q)$ for well-definedness, the value $f(G,p,q)$ is finite only if $p=q$.
Similarly, we have the following equation.
\begin{align*}
    h(G,p)=\begin{cases}
        +\infty & \text{ if } (p \text{ is even})\odot (v\in I_s\odot v\in I_t)=\FALSE\\
        p& \text{otherwise}
    \end{cases}
\end{align*}}

\subparagraph*{Operation 2.}
\revrev{As observed in \Cref{sec:operation2}, when identifying the roots, the number
of possible synchronized flips at the root clique can be chosen arbitrarily.
Accordingly, the recurrence in \Cref{sec:operation2} takes
$
\min_{q_1\in [0,n]} f(G_1,p,q_1)
$ and $
\min_{q_2\in [0,n]} f(G_2,p,q_2),
$
which correspond to $h(G_1,p)$ and $h(G_2,p)$, respectively.}

\subparagraph*{Operation 3.}
\revrev{Consider the case where $G$ is obtained from $G_1$ and $G_2$ by Operation~3.
We first explain the correspondence of $h(G_1, p_1)$ and $f(G_1,p_1,q_1)$ (resp. $h(G_2, p_2)$ and $f(G_2,p_2,q_2)$) for $p_1,p_2,q_1,q_2\in [0,n]$.
Recall that $f(G_1, p,q)$ is finite only when $p_1=q_1=p$ since $C(G_1)=\{r(G_1)\}$.
Hence $h(G_1,p)$ corresponds to the value $f(G_1, p,p)$.
Moreover, regarding $G_2$, the number of possible synchronized flips at the root clique can be chosen arbitrarily, by the same discussion as in \Cref{sec:operation3}.
Thus, $h(G_2, p_2)$ corresponds to $\min_{q_2\in[0,n]} f(G_2,p_2,q_2)=f^{*}(G_2,p_2)$.}

\revrev{Again, let $\sigma_{G_1}$ and $\sigma_{G_2}$ be shortest reconfiguration sequences witnessing
$f(G_1,p_1,q_1)$ and $f(G_2,p_2,q_2)$, respectively.
Assume that $s$ synchronized flips are possible when joining $\sigma_{G_1}$ and $\sigma_{G_2}$ into a single reconfiguration sequence $\sigma_G=\langle I_0,I_1,\dots,I_\ell\rangle$ of $G$, but that $\sigma_G$ contains fewer than $s$ synchronized flips in $C(G)$.
Then at least one possible synchronized flip is not performed, and hence there
exist three consecutive independent sets $I_i,I_{i+1},I_{i+2}$ in the sequence
that satisfy 
$|I_i\cap C(G)|=1$, $|I_{i+1}\cap C(G)|=0$, and $|I_{i+2}\cap C(G)|=1$.
Moreover, we have either $I_i\cap C(G)=\{r(G_1)\}$ and $I_{i+2}\cap C(G)=\{r(G_2)\}$, or $I_i\cap C(G)=\{r(G_2)\}$ and $I_{i+2}\cap C(G)=\{r(G_1)\}$.
By the observation for synchronized flip, $I_i$ and $I_{i+2}$ are adjacent.
By bypassing $I_{i+1}$, we obtain a shorter reconfiguration sequence.
Therefore, whenever $s$ synchronized flips are possible, we may assume that
all $s$ of them are performed.}

\revrev{By modifying the equation in \Cref{sec:operation3}, we obtain the following equation.
\begin{align*}
    h(G,p)=\min(h_1(G,p), h_2(G,p), h_3(G,p)),
\end{align*} 
where
\begin{align*}
    &h_1(G,p)=\begin{cases}
         \min_{1\le p_2< p} \{h(G_1,p)+h(G_2,p_2)-p_2\}& \text{ if } r(G_2)\in I_s\\
         \min_{0\le p_2< p} \{h(G_1,p)+h(G_2,p_2)-p_2\}& \text{ otherwise. }
    \end{cases}\\
    &h_2(G,p)=\begin{cases}
         h(G_1,0)+h(G_2,0) & \text{ if } p=0\land r(G_1)\in I_s\\
         \min_{p< p_2\le n} \{h(G_1,p)+h(G_2,p_2)-p\}& \text{ otherwise. }
    \end{cases}\\
    &h_3(G,p)=\begin{cases}
         h(G_1,p)+h(G_2,p)-p& \text{ if } |C(G)\cap I_s|=1\\
         h(G_1, 0)+h(G_2, 0)&\text{ if } |C(G)\cap I_s|=0\land p=0\\
         h(G_1,p)+h(G_2,p)-p+1& \text{ otherwise. }
    \end{cases}
\end{align*}}

\revrev{\subparagraph*{Running time.}
For each rooted tree $G'$ that appears in $T(G)$, there are $O(n)$ values of $h(G',p)$.
It takes $O(1)$ time and $O(n)$ time for Operations 1 and 2, respectively.
To reduce the running time for Operation 3, for each rooted tree $G'$, we precompute the auxiliary functions
$
h'_1(G',p')=\min_{p\le p'}\{h(G',p)-p\}
$
and 
$
h'_2(G',p')=\min_{p\ge p'}\{h(G',p)\}.
$
They can be computed in $O(n)$ time using the recurrences
$
h'_1(G',p')=\min\{h'_1(G',p'-1),\,h(G', p')-p'\}
$
and
$
h'_2(G',p')=\min\{h'_2(G',p'+1),h(G',p')\},
$
by processing $p'$ in increasing and decreasing order, respectively.
Hence, the corresponding update for Operation~3 takes $O(1)$ time, and $O(n)$ time in total.
Overall, the total running time is $O(n^2)$.
This completes the proof of \Cref{thm:tree}.}
\fi

\begin{theorem}\label{thm:tree}
    \prb{SymISR} can be solved in $O(n^2)$ time on trees.
\end{theorem}


\subsection{Cographs}

In this section, we prove that \prb{SymISR} is solvable in polynomial time on cographs, also known as $P_4$-free graphs~\cite{Cograph:Corneil81}.
To this end, we work in a slightly more general graph class, defined later.
A graph $G = (V, E)$ is a \emph{cograph} if it can be constructed recursively according to the following rules:
\begin{enumerate}
	\item A graph consisting of a single vertex is a cograph.
	\item If $G_1$ and $G_2$ are cographs, then their disjoint union $G_1 \oplus G_2$ is also a cograph.
	\item If $G_1$ and $G_2$ are cographs, then their join $G_1 \otimes G_2$ is also a cograph.
\end{enumerate}

We now state the theorem for graphs of the form $G_1 \otimes G_2$, where $G_1$ and $G_2$ are arbitrary graphs.
\begin{theorem}\label{thm:join}
    Let $\mathcal{G}$ be the class of graphs of the form $G_1 \otimes G_2$, where $G_1$ and $G_2$ are arbitrary graphs.
    Then \prb{SymISR} can be solved in linear time for $\mathcal{G}$.
\end{theorem}

This class trivially contains all connected cographs.
Since disconnected cographs can be solved independently on each connected component, the theorem immediately yields tractability on cographs.

We prove the following lemma, which provides \Cref{thm:join} when combined with \Cref{obs:length1}.
\begin{lemma}\label{lem:join}
    Let $\mathcal{G}$ be a class of graphs of the form $G_1 \otimes G_2$, where $G_1$ and $G_2$ are arbitrary graphs.
    For a graph $G\in \mathcal{G}$ and two independent sets $I_s, I_t$ of $G$, the length of the shortest reconfiguration sequence from $I_s$ to $I_t$ is at most $2$.
\end{lemma}
\iftrue
\begin{proof}
\revrev{Let $G$ be a graph such that $G\in \mathcal{G}$, and define $V_1\coloneq V(G_1)$ and $V_2\coloneq V(G_2)$.
We begin with a simple observation that for every independent set $I$ of $G$, we have $I\subseteq V_1$ or $I\subseteq V_2$; if there are two vertices $v_1,v_2\in I$ such that $v_1\in V_1$ and $v_2\in V_2$, then there is a contradiction that $uv\in E(G)$ and that $I$ is an independent set.
To prove \Cref{lem:join}, it is sufficient to show that $I_s$ and $I_t$ are adjacent, or there is an independent set $I_1$ that is adjacent to both $I_s$ and $I_t$.}

\revrev{First, consider that $I_s$ and $I_t$ belong to the same side. 
Then we assume that $I_s\subseteq V_1$ and $I_t\subseteq V_1$ without loss of generality.
Then we choose $v\in V_2$ arbitrary and set $I_1\coloneq \{v\}$ as an independent set $I_1$.
We now check the connectivity of $G[I_s\bigtriangleup I_1]$ and $G[I_1\bigtriangleup I_t]$, to verify the adjacency.
Indeed, for every vertex $u\in V_1$, there is an edge $uv\in E(G)$ by the definition.
Thus, if there is a vertex $u\in I_s$, then $G[I_s\bigtriangleup I_1]$ is connected.
Moreover, if $I_s=\emptyset$, the graph $G[I_s\bigtriangleup I_1]$ consist of a single vertex $\{v\}$, thus it is connected.
The same discussion can be applied to $I_t$; we can conclude that the length of the shortest reconfiguration sequence from $I_s$ to $I_t$ is at most $2$.}

\revrev{Second, consider that $I_s$ and $I_t$ belong to different sides. 
Then we assume that $I_s\subseteq V_1$ and $I_t\subseteq V_2$ without loss of generality.
We can see that $G[I_s\bigtriangleup I_t]=G[I_s\cup I_t]$ forms a complete bipartite graph, which leads to the adjacency of $I_s$ and $I_t$.}
\end{proof}
\fi

\subsection{Bipartite Chain Graphs}\label{sec:chain}

Now we consider bipartite chain graphs. These are bipartite graphs $G=(V,E)$ with the partition classes $A, B$ and some linear ordering $<_A$ on $A$, such that for all $a_1,a_2\in A$, $a_1<_A a_2$ implies $N(a_1)\subseteq N(a_2)$. Such an ordering also exists for $B$. 
In this subsection, we will prove that \prb{SymISR} is polynomial-time solvable on bipartite chain graphs. We start with results on general bipartite graphs.

\begin{observation}\label{obs:chain_case_bound}
    Let $G=(V,E)$ be a bipartite graph with the classes $A,B$ and two independent sets $I_1,I_2$. If $I_1 $ and $ I_2$ are adjacent, then \rev{\emph{(1)} $I_1 \setminus I_2 \subseteq A$ and $I_2 \setminus  I_1 \subseteq B$, or \emph{(2)} $I_2 \setminus I_1 \subseteq A$ and $I_1 \setminus  I_2 \subseteq B$.}

\end{observation}
\iftrue
\begin{proof}
    \revrev{If $I_1 \subseteq I_2$ or $I_2 \subseteq I_1$, the claim clearly holds.
    Assume that there is a vertex $v_1 \in I_1 \setminus I_2$ and a vertex $v_2 \in I_2 \setminus I_1$. Without loss of generality, assume $v_1\in A$. Since $I_1 \bigtriangleup I_2$ is connected, there is a path $v_1=u_1,\ldots,u_{\ell} = u$ on $G[I_1 \bigtriangleup I_2]$ for every $u \in I_1 \bigtriangleup I_2$. As $I_1,I_2,A$ and $B$ are independent sets, $u_{2i-1}\in (I_1 \setminus I_2) \cap A$ for all $i \in \{1,\ldots, \lceil\frac{\ell}{2} \rceil\}$, while $u_{2i}\in (I_2 \setminus I_1) \cap B$ for all $i \in \{1,\ldots, \lfloor\frac{\ell}{2} \rfloor\}$. This proves the observation.}
\end{proof}
\fi

 Let $G=(V,E)$ be a connected bipartite chain graph with the classes $A=\{a_1,\ldots, a_{n_A} \},B=\{b_1,\ldots, b_{n_B} \}$ as well as the orderings $<_A$ and $<_B$. 
 Without loss of generality, $a_1 <_A\ldots<_A a_{n_A}$ and $b_1 <_B\ldots<_B b_{n_B}$. 
 Let $I_s$ and $I_t$ be the initial and target independent sets. 
By \Cref{obs:length1}, the case $k=1$ is trivially solvable. Now, we consider $k\geq 3$. 
Define the independent sets $S := \{a_{n_A}\}\cup (I_s \cap A) \subseteq A$ and $T := \{b_{n_B}\} \cup  (I_t \cap B)\subseteq B$. Since $G$ is connected, $N(a_{n_A})=B$ and $N(b_{n_B})=A$. 
This implies that $S\bigtriangleup I_s \subseteq \{a_{n_A}\} \cup (I_s \cap B)$ and $T\bigtriangleup I_t \subseteq \{b_{n_B}\} \cup (I_t \cap A)$ are  connected. 
Since $\{a_{n_A},b_{n_B}\}\in E$ and $V= N(a_{n_A})\cup N(b_{n_B})$, $S\bigtriangleup T$ is connected. 
Hence, $\langle I_s,S,T,I_t\rangle$ is a reconfiguration sequence. 
Hence, we only need to consider the case $k=2$; we now show the equivalent condition for the existence of a shortest reconfiguration sequence with length $2$.
The following lemma gives an equivalent condition.

\begin{lemma}\label{lem:chain}
    Let $G$ be a bipartite chain graph and $I_s$ and $I_t$ be two independent sets such that $G[I_s\bigtriangleup I_t]$ is not connected. The length of the shortest reconfiguration sequence is $2$ if and only if
    \begin{enumerate}
        \item there is an independent set $I$ such that $|I_s\bigtriangleup I |=1$ (resp. $|I\bigtriangleup I_t|=1$) and $G[I\bigtriangleup I_t]$ (resp. $G[I_s\bigtriangleup I]$) is connected, or
        \item $(I_s \cap B) \setminus N(a_{\max}) = (I_t \cap B) \setminus N(a_{\max})$ or $(I_s \cap A) \setminus N(b_{\max}) = (I_t \cap A) \setminus N(b_{\max})$, where $a_{\max} \coloneq \max_{<_A} A\cap (I_s \cup I_t)$ and $b_{\max} \coloneq \max_{<_B} B\cap (I_s \cup I_t)$.
    \end{enumerate}
\end{lemma}
\iftrue

\begin{proof}
\revrev{We assume that there is the shortest reconfiguration sequence $\langle I_s,I, I_t\rangle $ of length $2$ for a connected bipartite chain graph $G=(V,E)$ with the classes $A,B$. 
If $I_s\bigtriangleup I \subseteq A$ then $\vert I_s\bigtriangleup I\vert = 1$, since it is connected. 
These cases can be checked in polynomial time since there are only a linear number of possibilities, and afterwards, we only have to check if $I\bigtriangleup I_s$ is connected, which can be done in overall linear time. 
Thus, we can assume $(I_r \bigtriangleup I) \cap C \neq \emptyset$ for $C\in \{A,B\} $ and $r\in \{s,t\}$. 
To bound the cases, we use \Cref{obs:chain_case_bound}. }

\revrev{We start with the case where $I_s\setminus I$ and $ I\setminus I_t$ are subsets of the same class (without the loss of generality, $A$). 
If there is a $b\in I_t\setminus I$, then it has a neighbor in $I\setminus I_t$, since $G[I_t\bigtriangleup I]$ is connected.
Here, we have $ I\setminus I_t\subseteq I\cap A$, and $ I\cap A\subseteq I_s\cap A$, since $(I_s\setminus I)\cap B=\emptyset$.
Thus, we have $I\setminus I_t\subseteq I_s\cap A$.
Hence, $(I_s\setminus I_t)\cap A=((I_s\setminus I)\cup (I\setminus I_t))\cap A=(I_s\setminus I)\cup (I\setminus I_t)$ holds in this case.
Moreover, $I_s\setminus I, I\setminus I_t\subseteq A$ immediately implies $I_s\setminus I_t\subseteq A$, thus we have $I_s\setminus I_t=(I_s\setminus I)\cup (I\setminus I_t)$.
By a similar discussion, we have $I_t\setminus I_s=(I_s\setminus I)\cup (I\setminus I_t)$.
Here, we consider $(I_s \bigtriangleup I) \cup (I\bigtriangleup I_t)= (I_s\setminus I) \cup (I\setminus I_t) \cup (I_t\setminus I) \cup (I\setminus I_s)$, which is indeed equal to $I_s\bigtriangleup I_t$ by the discussion above.
Since $I \bigtriangleup I_t$ and $I \bigtriangleup I_s$ are connected, clearly $I_s\bigtriangleup I_t$ is also connected; a contradiction.}

\revrev{Now, we consider the case $I_s\setminus I, I_t\setminus I\subseteq A$.
Our goals is to prove that this case is possible if and only if $(I_s \cap B) \setminus N(a_{\max}) = (I_t \cap B) \setminus N(a_{\max})$, where $a_{\max} \coloneq \max_{<_A} A\cap (I_s \cup I_t)$. 
Without loss of generality let $a_{\max} \in I_s$. Further, define  $a'_{\max}=\max_{<_A} A\cap I_t$.}

\revrev{First assume $(I_s \cap B) \setminus N(a_{\max}) = (I_t \cap B) \setminus N(a_{\max})$. 
Define $I=(I_t \cap B) \cup N(a_{\max}')$. 
Hence, $ (I_t \cap A)\subseteq (I_t \cap A) \cup N(a_{\max}')= I\bigtriangleup I_t$. This is connected, since for $b'=\max_{<_B}B\in N(a_{\max}')$, $N(b') =A$ and $I_t \cap A \neq \emptyset$.
By the same argument $I\bigtriangleup I_s = (I_s \cap A) \cup N(a_{\max}')$ is also connected, since $b_{\max} \in N(a'_{\max})$. 
Therefore, $\langle I_s, I,I_t\rangle$ is a reconfiguration sequence of length exactly $2$.}  

\revrev{For the other direction, assume there exists $I\subseteq V$ with $\emptyset\subsetneq I_s\setminus I, I_t\setminus I\subseteq A$ such that $\langle I_s, I,I_t\rangle$ is a reconfiguration sequence. 
Hence, $I \cap A \subseteq (I_s \cap I_t) \cap A $ and $(I_s \cup I_t) \cap B \subseteq  I \cap B$, since $a_{\max}$ is maximum in $A\cap(I_s\cup I_t )$. 
Let $p,q \in \{s,t\}$ be two distinct indices. 
Since $G$ is a bipartite chain graph, the definition of $a_{\max}$ implies $N((I_q \setminus I) \cap A )\subseteq N(I_q\cap A) \subseteq N(a_{\max})$. Furthermore, the assumption $I_s\setminus I, I_t\setminus I\subseteq A$ together with \Cref{obs:chain_case_bound} implies $I_q \cap B \subseteq I\cap B$. In total,
\begin{align*}
    (I_p \cap B) \setminus (I_q \cup N(a_{\max})) &\subseteq (I \cap B) \setminus (I_q \cup N(I_q \cap A)) \subseteq ((I \setminus I_q) \cap B) \setminus N((I_q \setminus I) \cap A)\\ &= ((I \bigtriangleup I_q) \cap B) \setminus N((I_q \bigtriangleup I) \cap A) =\emptyset.
\end{align*}
The last equality must hold, since $I_q \bigtriangleup I$ needs to be connected. 
Thus, we have $(I_p \cap B) \setminus (I_q \cup N(a_{\max}))= \emptyset$: implying $(I_p\cap B)\setminus N(a_{\max})\subseteq I_q$, and thus $ (I_q\cup B)\setminus N(a_{\max})$.
By symmetry, the reverse direction also holds; we have $(I_s \cap B) \setminus N(a_{\max}) = (I_t \cap B) \setminus N(a_{\max})$, completing the proof.} 
\end{proof}
\fi

The constraint $(I_s \cap B) \setminus N(a_{\max}) = (I_t \cap B) \setminus N(a_{\max})$ can be checked in linear time. If  $|I_s\bigtriangleup I|=1$, we have to go through all vertices and check if they could be the vertex in $I_s\bigtriangleup I$, (resp. $I_t\bigtriangleup I$). These are $n$ possibilities for which we have to check if the graph $G[I_t\bigtriangleup I]$ (resp. $G[I_s\bigtriangleup I]$) is connected.  
This implies the following Theorem.

\begin{theorem}
    \prb{SymISR} can be solved in $O(n^2)$ time on bipartite chain graphs.
\end{theorem}

\section{When the Optimum meets the Trivial Upper Bound}

One might expect that sequentially flipping the connected components of $G[I_s\bigtriangleup I_t]$ yields a good approximation ratio.
While this algorithm results in a worst-case $\Omega(n)$ approximation (for instance the reduction in \Cref{sec:hard}), the length of a shortest reconfiguration sequence coincides with this upper bound when we restrict the classes to paths and cycles.
This is formalized in the following theorem.
\begin{theorem}\label{thm:paths}
    For paths and cycles, the minimum length of a reconfiguration sequence from $I_s$ to $I_t$ is exactly $cc(I_s \bigtriangleup I_t)$.
\end{theorem}
\iftrue
\begin{proof}
\revrev{We first show that, for every independent set $I$, if $I'$ is an independent set such that $G[I\bigtriangleup I']$ is connected, then
$c(I'\bigtriangleup I_t)\ge c(I\bigtriangleup I_t)-1$.}

\revrev{Fix such an independent set $I'$.
Let $V'=I\bigtriangleup I'$, and partition $V'$ into $V_1=V'\cap (I\bigtriangleup I_t)$ and $V_2=V'\setminus (I\bigtriangleup I_t)$.
Since $G[V']$ is connected and $G$ is a path or cycle, any two distinct connected components of $G[V_1]$ are separated by a connected component of $G[V_2]$ along the path or cycle.
Therefore, $c(V_1)\le c(V_2)+1$.}

\revrev{We now evaluate $c(I'\bigtriangleup I_t)$ by considering the vertices in $V_1$ and $V_2$ separately.
On $V_1$, the vertices of $I\bigtriangleup I_t$ disappear when we pass to $I'\bigtriangleup I_t$, since
$I' \cap V_1 = I_t \cap V_1$.
Thus, the number of connected components decreases by $c(V_1)$.
On $V_2$, the sets $I'$ and $I_t$ are disjoint, that is,
$I'\cap I_t\cap V_2=\emptyset$.
Hence, the number of connected components increases by $c(V_2)$.
Thus, we have $c(I'\bigtriangleup I_t)=c(I\bigtriangleup I_t)-c(V_1)+c(V_2)$.
Combining with $c(V_1) \leq c(V_2)+1$, we have $c(I'\bigtriangleup I_t)\ge c(I\bigtriangleup I_t)-1$.
Since $I'$ was arbitrary, this proves the above statement.}

\revrev{Therefore, along any reconfiguration sequence from $I_s$ to $I_t$, each flip decreases the value of
$c(I\bigtriangleup I_t)$ by at most one, where $I$ denotes the current independent set.
Since $c(I_t\bigtriangleup I_t)=0$, every reconfiguration sequence from $I_s$ to $I_t$ has length at least $c(I_s\bigtriangleup I_t)$.
On the other hand, by the discussion in \Cref{sec:prbdef}, there exists a reconfiguration sequence of length exactly $c(I_s\bigtriangleup I_t)$.
This completes the proof of \Cref{thm:paths}.}
\end{proof}
\fi
As an immediate consequence of \Cref{thm:paths}, we obtain the following corollary.
\begin{corollary}
    \prb{SymISR} can be solved in $O(n)$ time for paths and cycles.
\end{corollary}

\subparagraph*{Remark.}
We remark that our analysis for paths and cycles is tight in some sense.
Specifically, we now show that sequentially flipping the connected components of $G[I_s\bigtriangleup I_t]$ yields $\Omega(n)$ approximation, for (1) a caterpillar with maximum degree 3, and (2) a graph with maximum degree $3$, and contains only one cycle.
To this end, we show that there is at least one instance for both cases such that $c(I_s\bigtriangleup I_t)$ is $\Theta(n)$, whereas the optimal length is $2$.

(1). Let $G$ be the graph obtained from a path $P=(p_1,p_2,...,p_{2t})$ with $2t$ vertices by adding $t$ vertices $L=\{\ell_i\mid i\in [t]]\}$ and $t$ edges $\{\ell_ip_{2i}\mid i\in [t]\}$.
We set $I_s=\{p_{2i-1}\mid i\in [t]\}$, and $I_t=I_s\cup L$.
Then, the number $c(I_s\bigtriangleup I_t)$ is exactly $c(L)$, thus is equal to $t$.
Despite this, we now show that the optimal length is $2$.
If we let $I'=\{ {p_{2i}\mid i\in [t]} \}$, then both $G[I_s\bigtriangleup I']$ and $G[I'\bigtriangleup I_t]$ are connected.

(2). We add an edge $p_{1}p_{2t}$ to $G$, and set $I_s$ and $I_t$ in the same way as above.
The discussion is also the same as above; thus, it is clear that this results in $\Omega(n)$ approximation.




\section{Conclusion and Future Work}
We study the shortest reconfiguration problem for independent sets under the adjacency relation derived by the independent set polytope; that is, each reconfiguration step corresponds to traversing an edge of the polytope.
Our results establish hardness for both sparse and dense graph classes, and tractability for some subclasses of perfect graphs.
Moreover, for paths and cycles, we prove that the optimal length of the reconfiguration sequence meets the trivial upper bound, derived from the definition of the adjacency relation.

Several directions remain for future work.
First, it is natural to ask about the complexity of \prb{SymISR} on bounded-treewidth graphs.
Second, our inapproximability results still leave a gap to the trivial upper bound, and it would be interesting to narrow this gap.
Third, it is worth investigating whether our algorithmic result can be extended to broader graph classes, such as bipartite graphs and distance-hereditary graphs.

\subparagraph*{Acknowledgments}
\rev{This work was supported by Institute of Mathematics for Industry, Joint
Usage/Research Center in Kyushu University (FY2024 Workshop (II) ``Theory of
Combinatorial Reconfiguration and Beyond'' (2024a037)).
We would like to thank the organisers and participants for the inspiring atmosphere.}
\newpage

\bibliographystyle{plain}
\bibliography{reference}

@article{planarSAT/Pilz19,
  author       = {Alexander Pilz},
  title        = {Planar {3-SAT} with a Clause/Variable Cycle},
  journal      = {Discret. Math. Theor. Comput. Sci.},
  volume       = {21},
  number       = {3},
  year         = {2019},
  url          = {https://doi.org/10.23638/DMTCS-21-3-18},
  doi          = {10.23638/DMTCS-21-3-18},
  bibsource    = {dblp computer science bibliography, https://dblp.org}
}

@article{reconf/ItoDHPSUU11,
  author       = {Takehiro Ito and
                  Erik D. Demaine and
                  Nicholas J. A. Harvey and
                  Christos H. Papadimitriou and
                  Martha Sideri and
                  Ryuhei Uehara and
                  Yushi Uno},
  title        = {On the complexity of reconfiguration problems},
  journal      = {Theor. Comput. Sci.},
  volume       = {412},
  number       = {12-14},
  pages        = {1054--1065},
  year         = {2011},
  url          = {https://doi.org/10.1016/j.tcs.2010.12.005},
  doi          = {10.1016/J.TCS.2010.12.005},
  bibsource    = {dblp computer science bibliography, https://dblp.org}
}

@article{kTJ/SugaSTZ25,
  author       = {Tatsuhiro Suga and
                  Akira Suzuki and
                  Yuma Tamura and
                  Xiao Zhou},
  title        = {Changing induced subgraph isomorphisms under extended reconfiguration
                  rules},
  journal      = {Inf. Comput.},
  volume       = {307},
  pages        = {105367},
  year         = {2025},
  url          = {https://doi.org/10.1016/j.ic.2025.105367},
  doi          = {10.1016/J.IC.2025.105367},
  bibsource    = {dblp computer science bibliography, https://dblp.org}
}

@inproceedings{k1TJ/KristanS25,
  author       = {Jan Maty{\'{a}}s Kristan and
                  Jakub Svoboda},
  title        = {Reconfiguration Using Generalized Token Jumping},
  booktitle    = {19th International Conference
                  and Workshops on Algorithms and Computation},
  series       = {Lecture Notes in Computer Science},
  pages        = {244--265},
  publisher    = {Springer},
  year         = {2025},
  url          = {https://doi.org/10.1007/978-981-96-2845-2\_16},
  doi          = {10.1007/978-981-96-2845-2\_16},
  bibsource    = {dblp computer science bibliography, https://dblp.org}
}

@incollection{reconfsurvey/Heuvel13,
  author       = {Jan van den Heuvel},
  editor       = {Simon R. Blackburn and
                  Stefanie Gerke and
                  Mark Wildon},
  title        = {The complexity of change},
  booktitle    = {Surveys in Combinatorics 2013},
  series       = {London Mathematical Society Lecture Note Series},
  pages        = {127--160},
  publisher    = {Cambridge University Press},
  year         = {2013},
  url          = {https://doi.org/10.1017/CBO9781139506748.005},
  doi          = {10.1017/CBO9781139506748.005},
  bibsource    = {dblp computer science bibliography, https://dblp.org}
}

@article{reconfsurvey/BousquetMNS24,
  author       = {Nicolas Bousquet and
                  Amer E. Mouawad and
                  Naomi Nishimura and
                  Sebastian Siebertz},
  title        = {A survey on the parameterized complexity of reconfiguration problems},
  journal      = {Comput. Sci. Rev.},
  volume       = {53},
  pages        = {100663},
  year         = {2024},
  url          = {https://doi.org/10.1016/j.cosrev.2024.100663},
  doi          = {10.1016/J.COSREV.2024.100663},
  bibsource    = {dblp computer science bibliography, https://dblp.org}
}

@inproceedings{MAS/SternSFK0WLA0KB19,
  author       = {Roni Stern and
                  Nathan R. Sturtevant and
                  Ariel Felner and
                  Sven Koenig and
                  Hang Ma and
                  Thayne T. Walker and
                  Jiaoyang Li and
                  Dor Atzmon and
                  Liron Cohen and
                  T. K. Satish Kumar and
                  Roman Bart{\'{a}}k and
                  Eli Boyarski},
  title        = {Multi-Agent Pathfinding: Definitions, Variants, and Benchmarks},
  booktitle    = {Twelfth International Symposium on Combinatorial
                  Search ({SOCS} 2019)},
  pages        = {151--158},
  publisher    = {{AAAI} Press},
  year         = {2019},
  url          = {https://doi.org/10.1609/socs.v10i1.18510},
  doi          = {10.1609/SOCS.V10I1.18510},
  bibsource    = {dblp computer science bibliography, https://dblp.org}
}

@article{ISpolytope/CHVATAL1975138,
title = {On certain polytopes associated with graphs},
journal = {Journal of Combinatorial Theory, Series B},
volume = {18},
number = {2},
pages = {138-154},
year = {1975},
issn = {0095-8956},
doi = {https://doi.org/10.1016/0095-8956(75)90041-6},
url = {https://www.sciencedirect.com/science/article/pii/0095895675900416},
author = {V. Chvátal}
}

@article{reconf/HearnD05,
  author       = {Robert A. Hearn and
                  Erik D. Demaine},
  title        = {PSPACE-completeness of sliding-block puzzles and other problems through
                  the nondeterministic constraint logic model of computation},
  journal      = {Theor. Comput. Sci.},
  volume       = {343},
  number       = {1-2},
  pages        = {72--96},
  year         = {2005},
  url          = {https://doi.org/10.1016/j.tcs.2005.05.008},
  doi          = {10.1016/J.TCS.2005.05.008},
  bibsource    = {dblp computer science bibliography, https://dblp.org}
}

@article{ISR/Wrochna18,
  author       = {Marcin Wrochna},
  title        = {Reconfiguration in bounded bandwidth and tree-depth},
  journal      = {J. Comput. Syst. Sci.},
  volume       = {93},
  pages        = {1--10},
  year         = {2018},
  url          = {https://doi.org/10.1016/j.jcss.2017.11.003},
  doi          = {10.1016/J.JCSS.2017.11.003},
  bibsource    = {dblp computer science bibliography, https://dblp.org}
}

@article{ISRtree/DemaineDFHIOOUY15,
  author       = {Erik D. Demaine and
                  Martin L. Demaine and
                  Eli Fox{-}Epstein and
                  Duc A. Hoang and
                  Takehiro Ito and
                  Hirotaka Ono and
                  Yota Otachi and
                  Ryuhei Uehara and
                  Takeshi Yamada},
  title        = {Linear-time algorithm for sliding tokens on trees},
  journal      = {Theor. Comput. Sci.},
  volume       = {600},
  pages        = {132--142},
  year         = {2015},
  url          = {https://doi.org/10.1016/j.tcs.2015.07.037},
  doi          = {10.1016/J.TCS.2015.07.037},
  bibsource    = {dblp computer science bibliography, https://dblp.org}
}

@article{TSsplit/BelmonteKLMOS21,
  author       = {R{\'{e}}my Belmonte and
                  Eun Jung Kim and
                  Michael Lampis and
                  Valia Mitsou and
                  Yota Otachi and
                  Florian Sikora},
  title        = {Token Sliding on Split Graphs},
  journal      = {Theory Comput. Syst.},
  volume       = {65},
  number       = {4},
  pages        = {662--686},
  year         = {2021},
  url          = {https://doi.org/10.1007/s00224-020-09967-8},
  doi          = {10.1007/S00224-020-09967-8},
  bibsource    = {dblp computer science bibliography, https://dblp.org}
}

@article{ISRbipartite/LokshtanovM19,
  author       = {Daniel Lokshtanov and
                  Amer E. Mouawad},
  title        = {The Complexity of Independent Set Reconfiguration on Bipartite Graphs},
  journal      = {{ACM} Trans. Algorithms},
  volume       = {15},
  number       = {1},
  pages        = {7:1--7:19},
  year         = {2019},
  url          = {https://doi.org/10.1145/3280825},
  doi          = {10.1145/3280825},
  bibsource    = {dblp computer science bibliography, https://dblp.org}
}

@inproceedings{shortISR/BodlaenderGS21,
  author       = {Hans L. Bodlaender and
                  Carla Groenland and
                  C{\'{e}}line M. F. Swennenhuis},
  title        = {Parameterized Complexities of Dominating and Independent Set Reconfiguration},
  booktitle    = {16th International Symposium on Parameterized and Exact Computation},
  series       = {LIPIcs},
  pages        = {9:1--9:16},
  publisher    = {Schloss Dagstuhl},
  year         = {2021},
  url          = {https://doi.org/10.4230/LIPIcs.IPEC.2021.9},
  doi          = {10.4230/LIPICS.IPEC.2021.9},
  bibsource    = {dblp computer science bibliography, https://dblp.org}
}

@article{hopISR/HatanoKIIM26,
  author       = {Hiroki Hatano and
                  Naoki Kitamura and
                  Taisuke Izumi and
                  Takehiro Ito and
                  Toshimitsu Masuzawa},
  title        = {Independent set reconfiguration under bounded-hop token jumping},
  journal      = {Theor. Comput. Sci.},
  volume       = {1062},
  pages        = {115651},
  year         = {2026},
  url          = {https://doi.org/10.1016/j.tcs.2025.115651},
  doi          = {10.1016/J.TCS.2025.115651},
  bibsource    = {dblp computer science bibliography, https://dblp.org}
}

@book{CyganFKLMPPS15,
  author       = {Marek Cygan and
                  Fedor V. Fomin and
                  Lukasz Kowalik and
                  Daniel Lokshtanov and
                  D{\'{a}}niel Marx and
                  Marcin Pilipczuk and
                  Michal Pilipczuk and
                  Saket Saurabh},
  title        = {Parameterized Algorithms},
  publisher    = {Springer},
  year         = {2015},
  doi          = {10.1007/978-3-319-21275-3},
}

@article{PMR/CardinalS25,
  author       = {Jean Cardinal and
                  Raphael Steiner},
  title        = {Inapproximability of shortest paths on perfect matching polytopes},
  journal      = {Math. Program.},
  volume       = {210},
  number       = {1},
  pages        = {147--163},
  year         = {2025},
  url          = {https://doi.org/10.1007/s10107-023-02025-4},
  doi          = {10.1007/S10107-023-02025-4},
  bibsource    = {dblp computer science bibliography, https://dblp.org}
}

@article{PMR/ItoKKKO22,
  author       = {Takehiro Ito and
                  Naonori Kakimura and
                  Naoyuki Kamiyama and
                  Yusuke Kobayashi and
                  Yoshio Okamoto},
  title        = {Shortest Reconfiguration of Perfect Matchings via Alternating Cycles},
  journal      = {{SIAM} J. Discret. Math.},
  volume       = {36},
  number       = {2},
  pages        = {1102--1123},
  year         = {2022},
  url          = {https://doi.org/10.1137/20m1364370},
  doi          = {10.1137/20M1364370},
  bibsource    = {dblp computer science bibliography, https://dblp.org}
}

@inproceedings{LokshtanovMPRS13,
  author       = {Daniel Lokshtanov and
                  Neeldhara Misra and
                  Geevarghese Philip and
                  M. S. Ramanujan and
                  Saket Saurabh},
  title        = {Hardness of $r$-dominating set on Graphs of Diameter $(r + 1)$},
  booktitle    = {Parameterized and Exact Computation - 8th International Symposium},
  series       = {Lecture Notes in Computer Science},
  volume       = {8246},
  pages        = {255--267},
  publisher    = {Springer},
  year         = {2013},
  doi          = {10.1007/978-3-319-03898-8\_22},
}

@inproceedings{inapprox/DinurS14,
  author       = {Irit Dinur and
                  David Steurer},
  editor       = {David B. Shmoys},
  title        = {Analytical approach to parallel repetition},
  booktitle    = {Symposium on Theory of Computing, {STOC} 2014},
  pages        = {624--633},
  publisher    = {{ACM}},
  year         = {2014},
  url          = {https://doi.org/10.1145/2591796.2591884},
  doi          = {10.1145/2591796.2591884},
  bibsource    = {dblp computer science bibliography, https://dblp.org}
}

@article{Pivotrule/KM72,
  title={How good is the simplex algorithm?},
  author={Klee, Victor and Minty, George J},
  journal={Inequalities},
  volume={3},
  number={3},
  pages={159--175},
  year={1972}
}

@article{Pivotrule/LoeraKS22,
  author       = {Jes{\'{u}}s A. De Loera and
                  Sean Kafer and
                  Laura Sanit{\`{a}}},
  title        = {Pivot Rules for Circuit-Augmentation Algorithms in Linear Optimization},
  journal      = {{SIAM} J. Optim.},
  volume       = {32},
  number       = {3},
  pages        = {2156--2179},
  year         = {2022},
  url          = {https://doi.org/10.1137/21m1419994},
  doi          = {10.1137/21M1419994},
  bibsource    = {dblp computer science bibliography, https://dblp.org}
}

@article{polymatroids/CardinalS25,
  author       = {Jean Cardinal and
                  Raphael Steiner},
  title        = {Shortest paths on polymatroids and hypergraphic polytopes},
  journal      = {Comb. Theory},
  volume       = {5},
  number       = {3},
  year         = {2025},
  url          = {https://doi.org/10.5070/c65365551},
  doi          = {10.5070/C65365551},
  bibsource    = {dblp computer science bibliography, https://dblp.org}
}

@article{polytopes/cunha2026,
      title={Computing distances is {FPT} on graph associahedra and {W[2]-hard} on hypergraphic polytopes}, 
      author={Luís Felipe I. Cunha and Ignasi Sau and Uéverton S. Souza and Mario Valencia-Pabon},
      year={2026},
      eprint={2504.18338},
      journal= {CoRR},
      volume= {abs/2603.05482},
}

@article{polytopes/BlackS2026,
  author       = {Alexander E. Black and
                  Raphael Steiner},
  title        = {Finding Short Paths on Simple Polytopes},
  journal      = {CoRR},
  volume       = {abs/2603.05482},
  year         = {2026},
  url          = {https://doi.org/10.48550/arXiv.2603.05482},
  doi          = {10.48550/ARXIV.2603.05482},
  eprinttype   = {arXiv},
  eprint       = {2603.05482},
  bibsource    = {dblp computer science bibliography, https://dblp.org}
}

@article{Cograph:Corneil81,
  author       = {Derek G. Corneil and
                  H. Lerchs and
                  L. Stewart Burlingham},
  title        = {Complement reducible graphs},
  journal      = {Discret. Appl. Math.},
  volume       = {3},
  number       = {3},
  pages        = {163--174},
  year         = {1981},
  url          = {https://doi.org/10.1016/0166-218X(81)90013-5},
  doi          = {10.1016/0166-218X(81)90013-5},
  bibsource    = {dblp computer science bibliography, https://dblp.org}
}

@phdthesis{kann1992approximability,
  title={On the approximability of NP-complete optimization problems},
  author={Kann, Viggo},
  year={1992},
  school={Royal Institute of Technology Stockholm},
  type         = {PhD thesis}
}

@book{graphclasses,
author = {Brandstädt, Andreas and Le, Van Bang and Spinrad, Jeremy P.},
title = {Graph Classes: A Survey},
publisher = {Society for Industrial and Applied Mathematics},
year = {1999},
doi = {10.1137/1.9780898719796},
address = {},
edition   = {},
URL = {https://epubs.siam.org/doi/abs/10.1137/1.9780898719796},
eprint = {https://epubs.siam.org/doi/pdf/10.1137/1.9780898719796}
}

@article{polytope:Naddef89,
  author       = {Denis Naddef},
  title        = {The hirsch conjecture is true for (0, 1)-polytopes},
  journal      = {Math. Program.},
  volume       = {45},
  number       = {1-3},
  pages        = {109--110},
  year         = {1989},
  url          = {https://doi.org/10.1007/BF01589099},
  doi          = {10.1007/BF01589099},
  bibsource    = {dblp computer science bibliography, https://dblp.org}
}

@article{polytope:BlackLKS25,
  author       = {Alexander E. Black and
                  Jes{\'{u}}s A. De Loera and
                  Sean Kafer and
                  Laura Sanit{\`{a}}},
  title        = {On the Simplex Method for 0/1-Polytopes},
  journal      = {Math. Oper. Res.},
  volume       = {50},
  number       = {2},
  pages        = {1398--1420},
  year         = {2025},
  url          = {https://doi.org/10.1287/moor.2021.0345},
  doi          = {10.1287/MOOR.2021.0345},
  bibsource    = {dblp computer science bibliography, https://dblp.org}
}


\end{document}